\documentclass[11pt]{article}
\usepackage[margin=1in]{geometry}
\usepackage{amsthm,amssymb,amsmath}  
\usepackage{xspace}
\usepackage[shortlabels]{enumitem}
\usepackage[colorlinks]{hyperref}
\usepackage[capitalise]{cleveref}
\usepackage[T1]{fontenc}
\usepackage[utf8]{inputenc}
\usepackage{thmtools}
\usepackage{thm-restate}
\usepackage{authblk}
\usepackage[caption=false]{subfig}
\usepackage{graphics}
\usepackage{multirow}
\usepackage{mathtools}

\usepackage{todonotes}
\usepackage{tikz}
\usepackage[noend,noline,boxed]{algorithm2e}
\usepackage{microtype}
\SetAlCapSkip{0.5em}
\usetikzlibrary{calc,decorations,decorations.pathreplacing}
\usepackage{makecell}

  \theoremstyle{plain}
  \newtheorem{theorem}{Theorem}
  \newtheorem{lemma}[theorem]{Lemma}  
  \newtheorem{corollary}[theorem]{Corollary} 
    
  \newtheorem{fact}[theorem]{Fact}
  \newtheorem{observation}[theorem]{Observation}
  \theoremstyle{definition}
  \newtheorem{definition}[theorem]{Definition}
  
  \newtheorem{example}[theorem]{Example}
  \newtheorem{remark}[theorem]{Remark}
 
  \newtheorem{claim}{Claim}

\newcommand{\rright}{\textsf{right}}
\newcommand{\lleft}{\textsf{left}}

\newcommand{\Surplus}{\mathsf{Surplus}}
\newcommand{\Mispers}{\mathsf{Mispers}}

\newcommand{\E}{\mathbf{E}} 
\newcommand{\T}{\mathcal{T}}

  \def\dd{\mathinner{.\,.}}
  
  \newcommand{\floor}[1]{\left\lfloor #1 \right\rfloor}
  \newcommand{\ceil}[1]{\left\lceil #1 \right\rceil}
  
  \newcommand{\Oh}{\mathcal{O}}
  \newcommand{\cO}{\mathcal{O}}
  \newcommand{\PREF}{\textsf{PREF}\xspace} 
  
  \newcommand{\LCP}{\texttt{LCP}}
  \newcommand{\IPM}{\texttt{IPM}}
  \newcommand{\Extract}{\texttt{Extract}}
  \newcommand{\Access}{\texttt{Access}}
  \newcommand{\Length}{\texttt{Length}}

  \newcommand{\Misp}{\mathsf{Misp}}
  
  \newcommand{\rot}{\mathsf{rot}}
  \newcommand{\Occ}{\textsf{Occ}}
  \newcommand{\cyc}[2]{\mathsf{CircOcc}^{#1}_{#2}}
  \newcommand{\ancycoc}{\mathsf{Anchored}}
  \newcommand{\Output}{\mathsf{output}}
  
  \newcommand{\allkep}{All-$k$-Edit-PrefMatch\xspace}

  \def\pillar{{\tt PILLAR}\xspace}
  
  \newcommand{\CPMA}{\textsf{CPM}}
  
  \newcommand{\ED}{\mathit{ED}}
  \newcommand{\ed}{\delta_E}
  \newcommand{\Ham}{\delta_H}
  \newcommand{\ddist}{\delta_D}
\newcommand{\lt}{\mathsf{lt}}
\newcommand{\br}{\mathsf{br}}
\newcommand{\dist}{\mathsf{dist}}

\newcommand{\LPref}{\textsf{LPref}}
\newcommand{\LSuf}{\textsf{LSuf}}

\newcommand{\Chain}{\textsf{Chain}}
\newcommand{\F}{\mathcal{F}}

\newcommand{\defproblem}[3]{

\smallskip
\vspace*{.7mm}%
\noindent\fbox{
\begin{minipage}{0.96\textwidth}
\vspace*{.7mm}
\textbf{Problem} #1
\vspace*{.8mm}%

\noindent
{\bf{Input:}} #2
\vspace*{.7mm}%

\noindent
{\bf{Output:}} #3%
\vspace*{.5mm}%
\end{minipage}
}
\vspace*{.5mm}
}

  \newcommand{\fragment}{substring\xspace}
  \newcommand{\fragments}{substrings\xspace}

\setlist[enumerate]{nosep, topsep=1ex}
\setlist[itemize]{nosep, topsep=1ex}
\setlist[description]{nosep}
\title{Approximate Circular Pattern Matching}

\author[1]{Panagiotis Charalampopoulos}
\author[2]{Tomasz Kociumaka}
\author[3]{Jakub Radoszewski}
\author[4,5]{Solon~P.~Pissis}
\author[3]{Wojciech Rytter}
\author[3]{Tomasz Waleń}
\author[3]{Wiktor Zuba}

\affil[1]{King's College London, UK\\
    \texttt{p.charalampopoulos@kcl.ac.uk}}
\affil[2]{Max Planck Institute for Informatics, Saarbrücken, Germany\\
    \texttt{tomasz.kociumaka@mpi-inf.mpg.de}}
\affil[3]{University of Warsaw, Poland\\
    \texttt{$\{$jrad,rytter,walen,w.zuba$\}$@mimuw.edu.pl}}
\affil[4]{CWI, Amsterdam, The Netherlands\\
    \texttt{solon.pissis@cwi.nl}}
\affil[5]{Vrije Universiteit, Amsterdam, The Netherlands}
\date{\vspace{-5ex}}

\begin{document}

\maketitle
\begin{abstract}
We investigate the complexity of approximate \emph{circular pattern matching} (CPM, in short) under the \emph{Hamming} and \emph{edit} distances.
In this problem, we are given a length-$n$ text~$T$, a length-$m$ pattern $P$, and a positive integer threshold~$k$, and we are to report all starting positions (called occurrences) of \fragments of $T$ that are at distance at most $k$ from some cyclic rotation of~$P$.
The decision version of the problem asks to check if any such occurrence exists.
All previous results for approximate CPM were either average-case upper bounds or heuristics, with the exception of the work of Charalampopoulos et al.~[CKP$^+$, JCSS'21], who considered the Hamming distance only. 
For the reporting version of the approximate CPM problem under the Hamming distance, we improve upon the main algorithm of [CKP$^+$, JCSS'21] from $\Oh(n+(n/m) \cdot k^4)$ to $\Oh(n+(n/m) \cdot k^3)$ time; for the edit distance, we give an $\Oh(nk^2)$-time algorithm.
Notably, for the decision versions and wide parameter ranges, we give algorithms whose complexities are almost identical to the state-of-the-art for standard (i.e., non-circular) approximate pattern matching:
\begin{itemize}
\item For the decision version of the approximate CPM problem under the Hamming distance,
we obtain an $\Oh(n+(n/m) \cdot k^2 \log k / \log \log k)$-time algorithm, which works in $\Oh(n)$ time whenever $k = \Oh(\sqrt{m \log \log m / \log m})$.
In comparison, the fastest algorithm for the standard counterpart of the problem, by Chan et al.~[CGKKP, STOC’20], runs in $\Oh(n)$ time only for $k=\Oh(\sqrt{m})$.
\item For the decision version of the approximate CPM problem under the edit distance, the $\Oh(nk\log^2 k)$ time complexity of our algorithm near matches the $\Oh(nk)$ time complexity of the Landau--Vishkin algorithm [LV, J.~Algorithms'89] for approximate pattern matching under edit distance; the latter algorithm remains the fastest known for $k=\Omega(m^{2/5})$.
As a stepping stone, we propose an $\Oh(nk\log^2 k)$-time algorithm for solving the Longest Prefix $k'$-Edit Match problem, proposed by Landau et al.\ [LMS, SICOMP'98], for all $k' \in \{1,\dots,k\}$.
We also show a further application of said problem in computing approximate longest common substrings.
\end{itemize}

We give a conditional lower bound that suggests a polynomial separation between approximate CPM under the Hamming distance over the binary alphabet and its non-circular counterpart.
We also show that a strongly subquadratic-time algorithm for the decision version of approximate CPM under the edit distance would refute the Strong Exponential Time Hypothesis.
\end{abstract}

\section{Introduction}

Pattern matching is one of the most widely studied problems in computer science. Given two strings, a \emph{pattern} $P$ of length $m$ and a \emph{text} $T$ of length $n$, the task is to find all occurrences of $P$ in $T$.
In the standard setting, the matching relation between $P$ and the substrings of $T$ assumes that the leftmost and rightmost positions of the pattern are conceptually important.
In many real-world applications, however, any rotation (cyclic shift) of $P$ is a relevant pattern. For instance, in bioinformatics~\cite{BMCgenomicsAyad2017,DBLP:conf/wea/BartonIKPRV15,DBLP:journals/almob/GrossiIMPPRV16,DBLP:books/sp/17/IliopoulosPR17}, the position where a sequence starts can be totally arbitrary due to arbitrariness in the sequencing of a circular molecular structure or due to inconsistencies introduced into sequence databases as a result of different linearization standards~\cite{BMCgenomicsAyad2017}. 
In image processing~\cite{DBLP:journals/prl/AyadBP17,DBLP:journals/ivc/Palazon-GonzalezM12,DBLP:journals/prl/Palazon-Gonzalez15,DBLP:journals/pr/Palazon-GonzalezMV14}, the contours of a shape may be represented through a directional chain code; the latter can be interpreted as a cyclic sequence if the orientation of the image is not important~\cite{DBLP:journals/prl/AyadBP17}.

With such scenarios in mind, when matching a pattern $P$ against a text $T$, one is interested in computing all \fragments of $T$ that match some rotation of $P$. Let us introduce the necessary basic notation.
For integers $x, y$, by $[x \dd y]=[x \dd y+1)=(x-1 \dd y+1)$ we denote an integer interval (or, simply, interval) $\{x,x+1,\ldots,y\}$, which is empty if $x>y$.
The positions of a string $U$ are numbered from 0 to $|U|-1$, with $U[i]$ denoting the $i$-th letter, and $U[i \dd j]=U[i \dd j+1)$ denoting the substring $U[i] \cdots U[j]$, which is empty if $i>j$.
One can then consider the (exact) Circular Pattern Matching (CPM) problem in which, given a text $T$ of length $n$ and a pattern $P$ of length $m$,
we are to compute the set $\{\, i\in [0\dd n-m]\,:\, T[i \dd i+m)=P' \text{ for some rotation }P'\text{ of }P\,\}$.

A textbook solution for CPM works in $\Oh(n\log\sigma)$ time  (or $\Oh(n)$ time with randomization),
where~$\sigma$ is the alphabet size, using the suffix automaton of $P\cdot P$~\cite{lothaire_2005}.
There is also a simple deterministic $\Oh(n)$-time algorithm which we discuss in \cref{sec:CPM}.
Many practically fast algorithms for CPM also exist; see~\cite{DBLP:journals/cj/ChenHL14,DBLP:journals/jda/FredrikssonG09,DBLP:conf/icmmi/SusikGD13} and references therein.
For the indexing version of the CPM problem (searching a collection of circular patterns), see~\cite{DBLP:conf/cpm/HonKST13,DBLP:conf/isaac/HonLST11, DBLP:books/sp/17/IliopoulosPR17,DBLP:journals/corr/abs-2504-03394}.

As in the standard pattern matching setting, a single changed, surplus, or missing
letter in $P$ or in~$T$ may result in many occurrences being missed.
In bioinformatics, this may correspond to a single-nucleotide polymorphism; in image processing, this may correspond to data corruption. 
We consider two well-known metrics on strings: the Hamming distance $\Ham$ (the number of mismatches for two equal-length strings, otherwise equal to $\infty$) and the edit distance $\ed$ (the minimum number of insertions, deletions, and substitutions required to transform one string into the other).
For two strings $U$ and $V$, an integer $k>0$, and a string metric $d$, we write $U =_k^d V$ if $d(U,V) \le k$ and $U \approx_k^d V$ if there exists a rotation $U'$ of $U$ such that $U' =_k^d V$.
We define 
\[\cyc{d}{k}(P,T) := \{i\,:\, P \approx_k^d T[i \dd j) \ \text{for some}\ j \ge i\}\]
we call $\cyc{d}{k}(P,T)$ the set of \emph{circular $k$-mismatch ($k$-edit) occurrences of $P$ in $T$} if $d=\delta_H$ ($d=\delta_E$, respectively).
We may omit the $d$-superscript when it is clear from the context.
Next, we formally define four variants of $k$-approximate CPM; see \cref{fig:ex} for an example.

\defproblem{$k$-Approximate CPM: $k$-Mismatch CPM and $k$-Edit CPM}{
  A text $T$ of length $n$, a pattern $P$ of length $m$, a positive integer $k$, and a distance function~$d$: $d=\Ham$ for $k$-Mismatch CPM and $d=\ed$ for $k$-Edit CPM.
}{
  ({\bf Reporting}) $\cyc{d}{k}(P,T)$.\\
  \hspace*{1.3cm} ({\bf Decision})  An arbitrary position $i \in \cyc{d}{k}(P,T)$, if any exists.
}

\begin{figure}
\centering
\begin{tikzpicture}
  \tikzstyle{red}=[color=red!90!black]
  \tikzstyle{darkred}=[color=red!50!black]
  \tikzstyle{blue}=[color=blue!50!black]
  \tikzstyle{black}=[color=black]
  \tikzstyle{green}=[color=green!50!black]
  \definecolor{turq}{RGB}{74,223,208}
  \definecolor{pink}{RGB}{254,193,203}

  \begin{scope}[xshift=-3.5cm,yshift=-1cm]
  \node at (-0.3,0) [left, above] {$P=$};
  \draw [fill=green!30!white] (0.16, 0.07) rectangle (1.05, 0.38);
  \draw [fill=turq] (1.05, 0.07) rectangle (2.25, 0.38);
  \foreach \c/\s [count=\i from 0] in {a/green,b/green,c/green,b/blue,b/blue,b/blue,b/blue} {
    \node at (\i * 0.3 + 0.3, 0) [above, \s] {\tt \c};
    \node at (\i * 0.3 + 0.3, 0.1) [below] {\tiny \i};
  }
  \end{scope}
  
  \begin{scope}[yshift=-2cm]
  \node at (-0.3,0) [left, above] {$T=$};
  \draw [fill=green!30!white] (2.54, 0.07) rectangle (3.43, 0.38);
  \draw [fill=turq] (1.35, 0.07) rectangle (2.54, 0.38);
  \foreach \c/\s [count=\i from 0] in {a/black,a/black,c/black,c/black,b/blue,c/blue,b/blue,b/blue,a/green,b/green,b/green,b/black} {
    \node at (\i * 0.3 + 0.3, 0) [above, \s] {\tt \c};
    \node at (\i * 0.3 + 0.3, 0.1) [below] {\tiny \i};
  }
  \end{scope}

  \begin{scope}[yshift=-1cm, xshift=1.2cm]
  \node at (-0.7,-0.1) [above] {$\rot_3(P)=$};
  \draw [fill=green!30!white] (1.34, 0.07) rectangle (2.25, 0.38);
  \draw [fill=turq] (0.15, 0.07) rectangle (1.34, 0.38);
  \foreach \c/\s [count=\i from 0] in {b/blue,b/red,b/blue,b/blue,a/green,b/green,c/red} {
    \node at (\i * 0.3 + 0.3, 0) [above, \s] {\tt \c};
  }
  \foreach \ii [count=\i from 0] in {3,4,5,6,0,1,2} {
    \node at (\i * 0.3 + 0.3, 0.1) [below] {\tiny \ii};
  }
  \end{scope}

\begin{scope}[xshift=5.9cm]
  \begin{scope}[yshift=-2cm]
  \node at (-0.3,0) [left, above] {$T'=$};
  \begin{scope}[xshift=-0.3cm]
  \draw [fill=green!30!white] (2.54, 0.07) rectangle (3.73, 0.38);
  \draw [fill=turq] (1.35, 0.07) rectangle (2.54, 0.38);
  \end{scope}
  \foreach \c/\s [count=\i from 0] in {a/black,a/black,c/black,b/blue,b/blue,c/blue,b/blue,a/green,c/green,b/green,c/green,b/black} {
    \node at (\i * 0.3 + 0.3, 0) [above, \s] {\tt \c};
    \node at (\i * 0.3 + 0.3, 0.1) [below] {\tiny \i};
  }
  \end{scope}

  \begin{scope}[yshift=-1cm, xshift=0.9cm]
  \node at (-0.7,-0.1) [above] {$\rot_3(P)=$};
  \draw [fill=green!30!white] (1.34, 0.07) rectangle (2.55, 0.38);
  \draw [fill=turq] (0.15, 0.07) rectangle (1.34, 0.38);
  \foreach \c/\s [count=\i from 0] in {b/blue,b/blue,b/red,b/blue,a/green,a/green!30!white,b/green,c/green} {
    \node at (\i * 0.3 + 0.3, 0) [above, \s] {\tt \c};
  }
  \foreach \ii [count=\i from 0] in {3,4,5,6,0,\,,1,2} {
    \node at (\i * 0.3 + 0.3, 0.1) [below] {\tiny \ii};
  }
  \end{scope}
  \end{scope}
  
\end{tikzpicture}
\caption{Left: a pattern $P$. Middle: a circular 2-mismatch occurrence of the pattern $P$ at position 4 in a text $T$. Right: a circular 2-edit occurrence of pattern $P$ (with the same rotation) at position 3 in text $T'$; note that there is no circular 2-mismatch occurrence of $P$ in $T'$ at this position.}\label{fig:ex}
\end{figure}

\paragraph*{Upper bounds.}
A summary of the previous and state-of-the-art worst-case upper bounds on approximate CPM for strings over a polynomially-bounded integer alphabet (we omit ``polynomially-bounded'' henceforth) is provided in \cref{tab:res_ub}.
A relatively large body of work has been devoted to practically fast algorithms for Approximate CPM; see~\cite{DBLP:journals/prl/AyadBP17,DBLP:journals/bmcbi/AyadPR16,DBLP:conf/isbra/AzimIRS15,DBLP:journals/almob/BartonIP14,DBLP:conf/lata/BartonIP15,DBLP:journals/jea/FredrikssonN04,DBLP:journals/jea/HirvolaT17,DBLP:journals/tjs/HoOK18} and references therein.
Charalampopoulos et al.~\cite{DBLP:journals/jcss/Charalampopoulos21} were the first to provide worst-case efficient algorithms for $k$-Mismatch CPM: they showed an $\Oh(nk)$-time algorithm and an $\Oh(n+(n/m)k^4)$-time algorithm.
After the preliminary version of this work~\cite{DBLP:conf/esa/Charalampopoulos22}, algorithms working in $\Oh(n+(n/m) \cdot k^6)$ time and $\Oh(n+(n/m) \cdot k^5 \log^3 k)$ time for the reporting and decision versions, respectively, of $k$-Edit CPM were presented in \cite{stacs24}. We remark that, when $k = \omega(\sqrt[4]{m})$, our algorithms for $k$-Edit CPM remain faster than those given in \cite{stacs24}.

\renewcommand{\arraystretch}{1.2}
\begin{table}[htpb]
    \centering
    \begin{tabular}{c|c|c|c}
        String metric & Version & Time complexity & Reference \\\hline
        \multirow{4}{*}{Hamming distance} & \multirow{3}{*}{reporting} & $\Oh(nk)$ & \multirow{2}{*}{\cite{DBLP:journals/jcss/Charalampopoulos21}} \\\cline{3-3}
        & & \textcolor{white!60!black}{$\Oh(n+(n/m) \cdot k^4)$} & \\\cline{3-4}
        & & $\Oh(n+(n/m) \cdot k^3)$ & \multirow{2}{*}{\cref{thm:ham_standard}}
          \\\cline{2-3}
        & decision & $\Oh(n+(n/m) \cdot k^2 \log k / \log \log k)$  & 
          \\\cline{1-4}
         \multirow{4}{*}{Edit distance} & \multirow{2}{*}{reporting} & $\Oh(nk^2)$ & \cref{thm:ed_rep}
         \\\cline{3-4} 
        &  & $\Oh(n+(n/m) \cdot k^6)$ & \cite{stacs24}
         \\\cline{2-4} 
        & \multirow{2}{*}{decision} & $\Oh(nk\log^2 k)$ & \cref{thm:ed_dec}
         \\\cline{3-4} 
         & & \textcolor{white!60!black}{$\Oh(n+(n/m) \cdot k^5 \log^3 k)$} & \cite{stacs24}
         \\\cline{3-4} 
         & & $\Oh(n+(n/m) \cdot k^5 \log^2 k)$ & \cite{stacs24} and \cref{lem:forSTACS}
    \end{tabular}
    \caption{A summary of known upper bounds for $k$-Approximate CPM. The bounds subsumed by subsequent work are displayed in gray.}\label{tab:res_ub}
\end{table}

In \cref{sec:mismatch,sec:std}, we prove the following result. In the preliminary version of the paper~\cite{DBLP:conf/esa/Charalampopoulos22}, the time complexity of the reporting version of $k$-Mismatch CPM was $\Oh(n + (n/m) \cdot k^3 \log \log k)$.

\begin{restatable}{theorem}{hamstand}\label{thm:ham_standard}
The reporting and decision versions of $k$-Mismatch CPM for strings over an integer alphabet can be solved in time $\Oh(n + (n/m) \cdot k^3)$ and $\Oh(n + (n/m) \cdot k^2 \log k / \log\log k)$, respectively.
\end{restatable}

A proof of the following theorem, based on the classic Landau--Vishkin algorithm~\cite{DBLP:journals/jal/LandauV89}, is given in \cref{sec:edit}.

\begin{restatable}{theorem}{reportingEditCPM}\label{thm:ed_rep}
The reporting version of $k$-Edit CPM for strings over an integer alphabet can be solved in $\Oh(nk^2)$ time.
\end{restatable}

We reduce the decision version of $k$-Edit CPM to the following problem.

\defproblem{Longest Prefix $k$-Edit Match ($k$-Edit PrefMatch)}{
  A text $T$ of length $n$, a pattern $P$, and a positive integer $k$.
}{
  An array $\LPref_k[0 \dd n]$ such that $\LPref_k[i]$ is the length of the longest prefix of~$P$ that matches a prefix of $T[i \dd n)$ with at most $k$ edits.
}

\noindent This problem was introduced under the name Longest Prefix $k$-Approximate Match by Landau et al.~\cite{DBLP:journals/siamcomp/LandauMS98}.
(An analogous problem for the Hamming distance has also been studied~\cite{DBLP:conf/swat/KaplanPS06}.)
Specifically, we introduce a problem called \allkep that consists in solving $k'$-Edit PrefMatch for all $k' \in [0\dd k]$.
We show that the decision version of $k$-Edit CPM can be reduced to \allkep on the same pattern and text and on the reversed pattern and text. Landau et al.~\cite{DBLP:journals/siamcomp/LandauMS98} gave an $\Oh(nk)$-time solution to $k$-Edit PrefMatch, which yields an $\Oh(nk^2)$-time solution to \allkep. In \cref{sec:AllkPREFMATCH}, we show the following result for \allkep (\cref{thm:allkPREFMATCH}), which implies \cref{thm:ed_dec}.

\begin{restatable}{theorem}{allkPREFMATCH}\label{thm:allkPREFMATCH}
\allkep for strings over an integer alphabet can be solved in time $\Oh(nk\log^2 k)$.
\end{restatable}

\begin{restatable}{theorem}{decisionEditCPM}\label{thm:ed_dec}
The decision version of $k$-Edit CPM for strings over an integer alphabet can be solved in time $\Oh(nk\log^2 k)$.
\end{restatable}

Compared with the preliminary version of the paper~\cite{DBLP:conf/esa/Charalampopoulos22}, the complexities in the two theorems are smaller by a factor of $\Theta(\log k)$. As a consequence, we improve the complexity of the algorithm from \cite{stacs24} for the decision version of $k$-Edit CPM to $\Oh(n+(n/m) \cdot k^5 \log^2 k)$. (It suffices to replace \cite[Lemma 7]{stacs24} by \cref{lem:forSTACS}, which is a counterpart of \cref{thm:ed_dec} stated in the \pillar model.) Since the algorithm in \cite{stacs24} was stated in the \pillar model, this automatically improves upon the complexities of the decision version of $k$-Edit CPM in the internal, dynamic, and fully compressed settings that are discussed in \cite[Appendix C]{stacs24}.

The complexities of our algorithms for the decision versions of $k$-Mismatch and $k$-Edit CPM match, up to $\log^{\cO(1)}k$ factors, the complexities of some of the fastest known algorithms for pattern matching with up to $k$ mismatches~\cite{DBLP:conf/stoc/ChanGKKP20,DBLP:conf/focs/Charalampopoulos20,DBLP:conf/soda/CliffordFPSS16,DBLP:conf/icalp/GawrychowskiU18} and edits~\cite{DBLP:journals/jal/LandauV89}, respectively.

In~\cite{DBLP:journals/siamcomp/LandauMS98}, an algorithm for the simpler problem of computing, given two strings $U$ and $V$ each of length at most $n$, the rotation of $U$ with the minimum edit distance to $V$ is given.
The algorithm works in $\Oh(ne)$ time, where $e$ is the minimum edit distance achieved.

\paragraph*{Applications of \allkep}
Let us consider the $k$-Edit Longest Common Substring ($k$-Edit LCSubstring) problem in which we are given two strings $S$ and $T$ of total length~$n$ and are to compute the longest substring of $S$ that matches a substring of $T$ with up to $k$ edits.
Abboud~et~al.~\cite{DBLP:conf/soda/AbboudWY15} proposed a solution for this problem working in $k^{1.5}n^2/2^{\Omega\left(\sqrt{{\log n}/{k}}\right)}$ time.
Let us notice that $k$-Edit LCSubstring can be solved directly by using the $k$-Edit PrefMatch problem for patterns equal to suffixes of $S$ and the text $T$.
By the result of Landau et al.~\cite{DBLP:journals/siamcomp/LandauMS98}, this yields an $\Oh(kn^2)$-time solution for $k$-Edit LCSubstring.
Moreover, our \cref{thm:allkPREFMATCH} yields an almost equally efficient, $\Oh(kn^2 \log^2 k)$-time algorithm that solves the $k'$-Edit LCSubstring for all $k' \in [1\dd k]$.
Both discussed algorithms have better complexity than the algorithm of Abboud et al.~\cite{DBLP:conf/soda/AbboudWY15} when $k$ is polynomial in $n$.
The techniques behind \allkep have already found other applications in approximate pattern matching~\cite{CPM25}.

\paragraph*{Our conditional lower bounds.}
We reduce known problems to approximate CPM, as shown in \cref{tab:res_lb}.

\begin{table}[htpb]
    \centering
    \begin{tabular}{c|c|c|c}
         Problem & Conditioned on & Complexity & Reference \\\hline
         Mismatch-CPM (unbounded) & BJI & \makecell{$\Omega(n^{1.5-\varepsilon})$ for all const.~$\varepsilon>0$} & \cref{hard:mismatch} 
         \\ \hline
         $k$-Edit CPM (decision) & SETH & $\Omega(n^{2-\varepsilon})$ for all const.~$\varepsilon>0$  & \cref{hard:edit}
    \end{tabular}
    \caption{Our conditional lower bounds for approximate CPM for alphabets of constant size.}
    \label{tab:res_lb}
\end{table}

For the Hamming distance, we consider the Mismatch-CPM problem where the number of allowed mismatches is unbounded (see \cref{sec:hardness} for a precise definition).
The breakthrough construction of a Binary Jumbled Index (BJI) in $\Oh(n^{1.859})$ time~\cite{DBLP:conf/stoc/ChanL15} was very recently improved to $n^{1.5}\log^{\Oh(1)} n$ time~\cite{https://doi.org/10.48550/arxiv.2204.04500}.
We show that, for any constant $\varepsilon>0$, obtaining an $\Oh(n^{1.5-\varepsilon})$-time algorithm for Mismatch-CPM over the binary alphabet would require a further improvement to BJI. In contrast, a similar problem of (non-circular) pattern matching with mismatches admits a classic $\Oh(n\log m)$-time solution (using convolutions) for constant-size alphabets~\cite{10.5555/889566}, whereas the fastest known solutions for a general alphabet are the randomised $\cO(n \sqrt{m})$-time algorithm of Chan~et~al.~\cite{DBLP:conf/focs/Chan0WX23} and the deterministic $\cO(n \sqrt{m \log\log m})$-time algorithm of Jin and Xu~\cite{DBLP:journals/corr/abs-2403-20326}.

Our conditional lower bound for $k$-Edit CPM is based on the hardness of computing the edit distance of two binary strings~\cite{DBLP:journals/siamcomp/BackursI18,DBLP:conf/focs/BringmannK15}, conditioned on the Strong Exponential Time Hypothesis (SETH)~\cite{DBLP:journals/jcss/ImpagliazzoP01}. It implies conditional optimality of our algorithm for the decision version of $k$-Edit CPM for a general $k \le n$, up to a subpolynomial factor. 

\paragraph*{The \pillar model.}
Our algorithms work in the \pillar model, introduced in~\cite{DBLP:conf/focs/Charalampopoulos20} with the aim of unifying approximate pattern matching algorithms across different settings.
In this model, we assume that the following primitive \pillar operations
can be performed efficiently, where
the argument strings are represented as \fragments of strings in a given collection $\mathcal{X}$:
\begin{itemize}
\item $\Extract(S, \ell, r)$: Retrieve string $S[\ell\dd r]$.
\item $\LCP(S, T),\, \LCP_R(S, T)$: Compute the length of the longest common prefix/suffix of $S$ and $T$.
\item $\IPM(S, T)$: Assuming that $|T| \le 2|S|$, compute the starting positions of all exact occurrences of $S$ in $T$, expressed as an arithmetic sequence.
\item $\Access(S, i)$: Retrieve the letter $S[i]$.
 \item $\Length(S)$: Compute the length $|S|$ of the string $S$.
\end{itemize}
The running time of algorithms in this model can be expressed in terms of the number of primitive \pillar operations and, if asymptotically larger, any extra time required by the algorithms.

In \cref{sec:mismatch}, we obtain fast algorithms for $k$-Mismatch CPM in the \pillar model.
Then, in \cref{app:settings}, we derive efficient solutions for the standard, internal, dynamic, fully compressed, quantum, packed, and read-only settings based on known implementations of the \pillar model in these settings. In particular, in the standard setting, where the strings are given explicitly and are over an integer alphabet, all primitive \pillar operations can be performed in $\Oh(1)$ time after a linear-time preprocessing.

\section{\texorpdfstring{$k$-Mismatch}{k-Mismatch} CPM in the \pillar Model}\label{sec:mismatch}
For a string $S$ and an integer $x \in [0\dd  |S|)$, we denote a cyclic rotation $\rot_x(S)=S[x \dd |S|)\cdot S[0\dd x)$.
For $i \in [0 \dd |S|-m]$, we denote $S^{(i)}=S[i\dd i+m)$.
Also, we denote the set of standard (non-circular) \emph{$k$-mismatch occurrences}  of $P$ in $T$ by
\[\Occ_k(P,T) = \{i \in [0\dd n-m]\,:\,T^{(i)} =_k P\}.\]

Let $P=P_1P_2$, where $|P_1| = \lfloor m/2\rfloor$.
Each circular 
$k$-mismatch occurrence of $P$ in $T$ implies a standard
$k$-mismatch occurrence of at least one of $P_1$ and $P_2$.
Henceforth, we assume (without loss of generality, as the remaining case is symmetric) that it implies a $k$-mismatch occurrence of $P_1$.
Our goal is to consider all $k$-mismatch occurrences of $P_1$ in $T$ as \emph{anchors} for computing circular $k$-mismatch occurrences of $P$ in $T$.
We call $P_1$ \emph{the sample}. 

By the so-called standard trick, we will consider
$\Oh(n/m)$ substrings of $T$ of length $\Oh(m)$ and process each of them separately. 
We denote one such substring by $T'$.
All positions (occurrences) of the sample can be efficiently computed in such a small window $T'$ by the following result (applied for $P=P_1$ and $T=T'$).

\begin{theorem}[{\cite[Main Theorems 5 and 8\protect\footnote{When referring to
  statements of~\cite{DBLP:conf/focs/Charalampopoulos20}, we use their numbering in the full (arXiv) version of
  the paper.}]{DBLP:conf/focs/Charalampopoulos20}}]\label{mt8} 
Given a pattern $P$ and a text $T$ of length $|T|=\Oh(|P|)$, a representation of the set $\Occ_k(P,T)$ as $\cO(k^2)$ pairwise disjoint arithmetic progressions can be computed using $\Oh(k^2 \log \log k)$ time plus $\Oh(k^2)$ \pillar operations.
\end{theorem}

We need to extend the $k$-mismatch occurrences of $P_1$ in $T'$ into \fragments of $T$ which approximately match some rotation of $P$.
Let us consider only rotations of $P$ of the form $YP_1X$, where $P=P_1XY$.
We define the set of circular $k$-mismatch occurrences $i$ of $P$ in $T$ that imply a $k$-mismatch standard occurrence of a prefix $P'$ of $P$ in a \fragment $T'=T[a \dd b]$ as follows:
    \[\ancycoc_k(T,T',P,P') = 
    \{\, i\,:\, T^{(i)} =_k Y P' X,\,P=P' X Y,\; i+|Y|-a\in \Occ_k(P',T')\}.\]
We say that such occurrences are \emph{anchored} in $(P',T')$. Our algorithms compute a superset of $\ancycoc_k(T,T',P,P_1)$ that may also contain other circular $k$-mismatch occurrences of $P$ in $T$ (some of the ones that contain a $k$-mismatch occurrence of $P_2$).

\subsection{Few \texorpdfstring{$k$-mismatch}{k-mismatch} Occurrences of the Sample}\label{subsec:few}

We first assume that $|
\Occ_k(P_1,T')|=\Oh(k)$.
Let us denote by \textsc{PairMatch}$_k(T, i, P, j)$ the set of all circular $k$-mismatch occurrences
of $P$ in $T$ such that position $i$ in $T$ is aligned with 
position $j$ in $P$; see \cref{fig:ex_pairmatch}. Formally,
\[\textsc{PairMatch}_k(T, i, P, j)= \{ p \in [i-m+1 \dd i]\,:\,T^{(p)} =_k \rot_x(P),\, i-p \equiv j-x \pmod{m}\}.\] 
In particular, \textsc{PairMatch}$_k(T, i, P, 0)$ is the set of circular $k$-mismatch occurrences of $P$ such that the leftmost position of $P$ is aligned with position $i$ in $T$.

\begin{figure}[htpb]
\centering
\begin{tikzpicture}
  \tikzstyle{red}=[color=red!90!black]
  \tikzstyle{darkred}=[color=red!50!black]
  \tikzstyle{blue}=[color=blue!50!black]
  \tikzstyle{black}=[color=black]
  \tikzstyle{green}=[color=green!50!black]
  \definecolor{turq}{RGB}{74,223,208}
  \definecolor{pink}{RGB}{254,193,203}

  \begin{scope}[xshift=-5cm,yshift=-2cm]
  \node at (-0.3,0) [left, above] {$P=$};
  \foreach \c/\s [count=\i from 0] in {a/black,b/black,c/black,a/black,b/black,a/black,c/black} {
    \node at (\i * 0.3 + 0.3, 0) [above, \s] {\tt \c};
    \node at (\i * 0.3 + 0.3, 0.1) [below] {\tiny \i};
  }
  \end{scope}
  
  \begin{scope}[yshift=-2cm]
  \node at (-0.3,0) [left, above] {$T=$};
  \foreach \c/\s [count=\i from 0] in {c/black,a/black,b/black,b/black,a/black,b/black,b/black,c/black,a/black,d/black,c/black,a/black,a/black} {
    \node at (\i * 0.3 + 0.3, 0) [above, \s] {\tt \c};
  }
  \foreach \ii [count=\i from 0] in {0,1,2,3,4,5,6,7,{\bf8},9,10,11,12} {
    \node at (\i * 0.3 + 0.3, 0.1) [below] {\tiny \ii};
  }
  \end{scope}

  \begin{scope}[yshift=-1cm, xshift=1.2cm]
  \node at (-0.7,-0.1) [above] {$\rot_3(P)=$};
  \draw [fill=green!30!white] (1.34, 0.07) rectangle (2.25, 0.38);
  \draw [fill=turq] (0.15, 0.07) rectangle (1.34, 0.38);
  \foreach \c/\s [count=\i from 0] in {a/blue,b/blue,a/red,c/blue,a/green,b/red,c/green} {
    \node at (\i * 0.3 + 0.3, 0) [above, \s] {\tt \c};
  }
  \foreach \ii [count=\i from 0] in {3,4,5,6,{\bf0},1,2} {
    \node at (\i * 0.3 + 0.3, 0.1) [below] {\tiny \ii};
  }
  \end{scope}

    \begin{scope}[yshift=0cm, xshift=1.5cm]
  \node at (-0.7,-0.1) [above] {$\rot_4(P)=$};
  \draw [fill=green!30!white] (1.04, 0.07) rectangle (2.25, 0.38);
  \draw [fill=turq] (0.15, 0.07) rectangle (1.04, 0.38);
  \foreach \c/\s [count=\i from 0] in {b/blue,a/red,c/blue,a/green,b/red,c/green,a/green} {
    \node at (\i * 0.3 + 0.3, 0) [above, \s] {\tt \c};
  }
  \foreach \ii [count=\i from 0] in {4,5,6,{\bf0},1,2,3} {
    \node at (\i * 0.3 + 0.3, 0.1) [below] {\tiny \ii};
  }
  \end{scope}

    \begin{scope}[yshift=-3cm, xshift=0.6cm]
  \node at (-0.7,-0.1) [above] {$\rot_1(P)=$};
  \draw [fill=green!30!white] (1.94, 0.07) rectangle (2.25, 0.38);
  \draw [fill=turq] (0.15, 0.07) rectangle (1.94, 0.38);
  \foreach \c/\s [count=\i from 0] in {b/blue,c/red,a/blue,b/blue,a/red,c/blue,a/green} {
    \node at (\i * 0.3 + 0.3, 0) [above, \s] {\tt \c};
  }
  \foreach \ii [count=\i from 0] in {1,2,3,4,5,6,{\bf0}} {
    \node at (\i * 0.3 + 0.3, 0.1) [below] {\tiny \ii};
  }
  \end{scope}

\end{tikzpicture}
\caption{$\textsc{PairMatch}_2(T, 8, P, 0)\,=\,\{2,4,5\}$.}\label{fig:ex_pairmatch}
\end{figure}

The following lemma was shown in~\cite{DBLP:journals/jcss/Charalampopoulos21} without explicitly mentioning the \pillar model.

\begin{lemma}[{\cite[see Lemma 10]{DBLP:journals/jcss/Charalampopoulos21}}]\label{lem:PairMatch}
For any given $k,i,j$, the set \textsc{PairMatch}$_k(T, i, P, j)$,
represented as a union of $\Oh(k)$ intervals, can be computed in $\Oh(k)$ time in the \pillar model.
\end{lemma}

The next lemma shows how to compute $\ancycoc_k(T,T',P,P_1)$ if $P_1$ has few occurrences in $T'$.

\begin{lemma}\label{lem:few}
Let $P'$ be a prefix of $P$ and $T'$ be a \fragment of $T$.
Given $A=\Occ_k(P',T')$ of size $\Oh(k)$, a superset of $\ancycoc_k(T,T',P,P')$, represented as a union of $\Oh(k^2)$ intervals, can be computed in $\Oh(k^2)$ time in the \pillar model. The superset is contained in $\cyc{}{k}(P,T)$, i.e., it consists only of circular $k$-mismatch occurrences of $P$ in~$T$.
\end{lemma}
\begin{proof}
Suppose that $T'=T[a \dd b]$. We then have
\[\ancycoc_k(T,T',P,P')\ \subseteq\ \bigcup_{i\in A}\, \textsc{PairMatch}_k(T, i+a, P, 0)\ \subseteq\ \cyc{}{k}(P,T).\]
By the assumption $|A| = \Oh(k)$, the result, represented as a union of $\Oh(k^2)$ intervals, can be computed in $\Oh(k^2)$ time in the \pillar model using the algorithm of \cref{lem:PairMatch}.
\end{proof}

\subsection{Many \texorpdfstring{$k$-mismatch}{k-mismatch} Occurrences of the Sample}
We henceforth assume that $|\Occ_k(P_1,T')| > 864k$. By~\cite{DBLP:conf/focs/Charalampopoulos20}, this yields
approximate periodicity in both $P_1$ and the portion of $T'$ spanned by the $k$-mismatch occurrences of $P_1$.
We show that, except for $\Oh(k^2)$ intervals of circular $k$-mismatch occurrences of $P$ that we can compute using $\cO(k)$ calls to $\textsc{PairMatch}_k$ as in \cref{subsec:few}, our problem reduces to matching length-$m$ substrings of an approximately periodic substring $U$ of $T$ and of an approximately periodic substring $V$ of $P^2$ with aligned approximate periodicity.
Testing a match of two length-$m$ substrings for $U$ and $V$ reduces to a test only involving positions breaking periodicity, called \emph{misperiods}, since the equality on other positions is guaranteed due to common approximate periodicities. The number of misperiods in~$U$ and $V$ is only $\Oh(k)$, so the complexity of our algorithms, in terms of \pillar operations, only depends on $k$.

By $S^p$ we denote the concatenation of $p$ copies of a string $S$. A string $Q$ is called primitive if there is no string $B$ and integer $p\ge 2$ such that $Q=B^p$. We say that a string $U$ is \emph{almost $Q$-periodic} if $U$ is at Hamming distance $\Oh(k)$ from a prefix of $Q^\infty$. We introduce the following problem; the notations are illustrated in \cref{fig:setting}.\footnote{\textsc{PeriodicSubMatch} was generalized to edit distance in \cite{stacs24}.}

\defproblem{\textsc{PeriodicSubMatch}$(U,V)$}{
A primitive string $Q$, integers $m$, $k$, $q$, $r$, $\alpha_U$, $\beta_U$, $\alpha_V$, $\beta_V$, and strings $U$, $V$ such that:
\begin{itemize}
\item $m \le |U|,|V| \le 2 m$, $q=|Q|$, $r \in [0 \dd q)$,
\item $U=T[\alpha_U \dd \beta_U]$ and $V\,=\,P^2[\alpha_V\dd\beta_V]$,
\item $U$ is almost $Q$-periodic and $V$ is almost $Q'$-periodic, where $Q'=\rot_r(Q)$.
\end{itemize}
}{
$\bigcup_{x \in [0 \dd |V|-m]}\{\,p\in \Occ_k(V^{(x)},U)\;:\;
p-x\equiv r \pmod{q}\,\}$.
}

\begin{figure}[htpb]
    \centering
\begin{tikzpicture}[xscale=0.25,yscale=0.5]
\fill[green!20!white,xshift=-1cm] (12.5,0) rectangle (35.5,1);
\draw (12,0) node[below] {$p$};
\foreach \x/\c in {0/a,1/b, 3/d,4/e,5/f,6/g,7/h, 8/a,9/b, 10/c, 12/e,13/f,14/g,15/h, 16/a,17/b,18/c,  20/e,21/f,22/g,23/h, 24/a,25/b,26/c,27/d,28/e,29/f,30/g,31/h,33/b,34/c,35/d}{
  \draw (\x,-0.1) node[above] {\small \vphantom{\texttt{g}}{\texttt{\c}}};
}
\foreach \x/\c in {2/t,11/h,19/z,32/x}{
  \draw (\x,-0.1) node[above,red] {\small \vphantom{\texttt{g}}{\texttt{\c}}};
}
\draw (-0.5,0) rectangle (35.5,1);
\draw (-0.5,0.5) node[left] {$U$};
\begin{scope}
\clip (-0.5,1) rectangle (35.5,3);
\draw[xshift=-0.5cm] (0,1) .. controls (3,2) and (5,2) .. (8,1);
\draw[xshift=7.5cm] (0,1) .. controls (3,2) and (5,2) .. (8,1);
\draw[xshift=15.5cm] (0,1) .. controls (3,2) and (5,2) .. (8,1);
\draw[xshift=23.5cm] (0,1) .. controls (3,2) and (5,2) .. node[above] {$Q$} (8,1);
\draw[xshift=31.5cm] (0,1) .. controls (3,2) and (5,2) .. (8,1);
\end{scope}

\begin{scope}[yshift=-5.2cm,xshift=0cm]
\draw (12,0) node[below] {$x$};
\fill[green!20!white,xshift=-1cm] (12.5,0) rectangle (35.5,1);
\foreach \x/\c in {6/g,7/h, 8/a,9/b,10/c,11/d,12/e,13/f,14/g,15/h, 16/a,17/b,18/c, 20/e,21/f,22/g,23/h, 24/a,25/b, 27/d,28/e,29/f,30/g,31/h,
  33/b,34/c,35/d, 37/f
}{
  \draw (\x,-0.1) node[above] {\small \vphantom{\texttt{g}}{\texttt{\c}}};
}
\foreach \x/\c in {19/z,26/y,32/w,36/s}{
  \draw (\x,-0.1) node[above,red] {\small \vphantom{\texttt{g}}{\texttt{\c}}};
}
\draw (5.5,0) rectangle (37.5,1);
\draw (5.5,0.5) node[left] {$V$};
\begin{scope}
\clip (5.5,1) rectangle (37.5,3);
\draw[xshift=-0.5cm] (0,1) .. controls (3,2) and (5,2) .. (8,1);
\draw[xshift=7.5cm] (0,1) .. controls (3,2) and (5,2) .. (8,1);
\draw[xshift=15.5cm] (0,1) .. controls (3,2) and (5,2) .. (8,1);
\draw[xshift=23.5cm] (0,1) .. controls (3,2) and (5,2) .. (8,1);
\draw[xshift=31.5cm] (0,1) .. controls (3,2) and (5,2) .. (8,1);
\end{scope}
\begin{scope}[yshift=-1cm]
\draw[densely dashed] (15.5,-1) rectangle (40.5,0);
\draw (27.5,-0.5) node {$P$};
\draw[xshift=-24cm,densely dashed] (39.5,-1) -- (14.5,-1) -- (14.5,0) -- (39.5,0);
\draw[xshift=-27cm] (27.5,-0.5) node {$P$};
\draw (6,-1) node[blue!50!black,below] {$\alpha_V$};
\draw[densely dotted] (5.5,-1) -- (5.5,1);
\draw (37,-1) node[blue!50!black,below] {$\beta_V$};
\draw[densely dotted] (37.5,-1) -- (37.5,1);
\end{scope}

\draw[xshift=-8cm,latex-latex] (7.5,2) -- node[above] {\small $p-x=r$} (13.5,2);
\draw (5.5,3.0) -- (5.5,3.2) -- node[above] {$Q'$} (13.5,3.2) -- (13.5,3.0);
\draw[densely dotted] (5.5,3.0) -- (5.5,1)  (13.5,3.0) -- (13.5,1); 

\end{scope}

\end{tikzpicture}

    \caption{\textsc{PeriodicSubMatch} with $m=24$, $k=2$, $q=8$, $r=5$, $\alpha_V=14$, $\beta_V=45$. Mismatches with respect to the approximate periodicity are marked in red. Green rectangles show that $U^{(p)}=_2 V^{(x)}$. In this example we have $p-x=r$; in general, $p-x \equiv r \pmod{q}$.
    }\label{fig:setting}
\end{figure}

Our goal is to reduce $k$-Mismatch CPM to \textsc{PeriodicSubMatch}. To this end, we use the notion of repetitive regions from~\cite{DBLP:conf/focs/Charalampopoulos20}.

\begin{definition}\label{def:repreg}
A \emph{$k$-repetitive region} in a string $S$ of length $m$ is a substring $R$ of $S$ of length $|R| \ge 3m/8$ for which there exists a primitive string $Q$ such that 
\[|Q| \le m/(128k)\ \text{and}\ \Ham(R,Q^{\infty}[0 \dd |R|)) = \lceil 8k|R|/m \rceil.\]
\end{definition}

\noindent
The following lemma is a simplified version of a lemma from~\cite{DBLP:conf/focs/Charalampopoulos20} with one repetitive region.

\begin{lemma}[see~{\cite[Lemma 3.11]{DBLP:conf/focs/Charalampopoulos20}}]\label{repreg}
Given a pattern $P$ of length $m$, a text $T$ of length $n \le \frac32 m$, and a positive integer threshold $k \le m$, if $P$ contains a $k$-repetitive region, then $|\Occ_k(P,T)| = \Oh(k)$.
\end{lemma}

Intuitively, in our main lemma in this section (\cref{lem:toPSM}), 
we show that if $P_1$ has many $k$-mismatch occurrences in $T'$, then each circular $k$-mismatch occurrence of $P$ in~$T$ that is an element of $\ancycoc_k(T,T',P,P_1)$ is: either (a) computed in an instance of \textsc{PeriodicSubMatch}; or (b) implies a $k$-mismatch occurrence of one of at most two $k$-repetitive regions of $P_2P_1P_2$.
The occurrences of the second type can be computed using $\cO(k)$ calls to $\textsc{PairMatch}_k$ by viewing $k$-repetitive regions as samples and applying \cref{repreg} to bound the number of $k$-mismatch occurrences.

In the proof of \cref{lem:toPSM}, we use the following theorem from~\cite{DBLP:conf/focs/Charalampopoulos20}; part \eqref{it1} is a consequence of~\cite[Theorem 3.1]{DBLP:conf/focs/Charalampopoulos20} (existence of $Q$) and~\cite[Lemmas 3.8, 3.11, 3.14, and 4.4]{DBLP:conf/focs/Charalampopoulos20} (computation of $Q$), whereas the remaining parts are due to~\cite[Main Theorem 5]{DBLP:conf/focs/Charalampopoulos20}.
\begin{theorem}[\cite{DBLP:conf/focs/Charalampopoulos20}]\label{mt4}
  Assume that we are given a pattern $P$ of length $m$, a text $T$ of length $n \le \frac32 m$,
  and a positive integer threshold $k \le m$. If $|\Occ_k(P,T)|>864 k$ and $T$ starts and ends with $k$-mismatch occurrences of $P$, that is, $0,n-m \in \Occ_k(P,T)$, then:
  \begin{enumerate}
  \item\label{it1} there is a primitive substring $Q$ of $P$ satisfying $|Q| \le m/(128k)$ and $\Ham(P, Q ^\infty [0\dd |P|)) < 2k$, and such a substring, if it exists,
  can be computed in $\Oh(k)$ time in the \pillar model;
  \item each element of $\Occ_k(P,T)$ is a multiple of $|Q|$;
  \item $\Ham(T,Q^\infty [0\dd n)) \le 6k$;
  \item $\Occ_k(P,T)$ can be decomposed into $\Oh(k^2)$ arithmetic progressions with difference~$|Q|$.
  \end{enumerate}
\end{theorem}

We also use the next lemma that adapts the kangaroo jumping technique~\cite{DBLP:journals/tcs/GalilG87,DBLP:journals/tcs/LandauV86}.

\begin{lemma}[{\cite[Lemma 4.1]{DBLP:conf/focs/Charalampopoulos20}}]\label{lem:misoph}
    Let $S$ and $Q$ denote strings.
    There is a generator {\tt MismGenerator($S$, $Q$)} ({\tt MismGenerator$^R$($S$,~$Q$)}) that, in the $k$-th call, returns in $\Oh(1)$ time in the \pillar model the length of the longest prefix (suffix, respectively) $S'$ of $S$ such that $\Ham(S', W) \le k$, where $W$ denotes the prefix (suffix, respectively) $W$ of~$Q^{|S|}$ of length $|S'|$.
\end{lemma}

\begin{lemma}\label{lem:toPSM}
The $k$-Mismatch CPM problem can be reduced using $\Oh((n/m)k^2\log \log k)$ time plus $\Oh((n/m)k^2)$  \pillar operations to $\Oh(n/m)$ instances of the \textsc{PeriodicSubMatch} problem. The output to $k$-Mismatch CPM is a union of the outputs of the \textsc{PeriodicSubMatch} instances and $\Oh((n/m)k^2)$ intervals.
Each circular $k$-mismatch occurrence can be returned by $\Oh(1)$ instances of \textsc{PeriodicSubMatch} and $\Oh(1)$ intervals.
\end{lemma}
\begin{proof}
By symmetry, it suffices to describe the case where $P_1$ is the sample.
We use the standard trick to cover $T$ by a collection of \fragments $T'$, each of length $\floor{\frac34 m}$, starting at positions divisible by $\floor{\frac14 m}$.
We note that each circular $k$-mismatch occurrence of $P$ in $T$ can be present in $\Oh(1)$ sets $\ancycoc_k(T,T',P,P_1)$.
For each of the $\Oh(n/m)$ substrings $T'$, we obtain potentially an instance of \textsc{PeriodicSubMatch} and a set of $\cO(k^2)$ intervals, which we can sort in $\cO(k^2 \log\log k)$ time, so that we can compute their union represented as a family of pairwise disjoint intervals in $\Oh((n/m)k^2\log \log k)$ total time.
In particular, for each substring $T'$, we compute (a superset of) $\ancycoc_k(T,T',P,P_1)$.

We use the algorithm of \cref{mt8} to compute a representation of the set $A=\Occ(P_1,T')$ using $\Oh(k^2 \log \log k)$ time plus $\Oh(k^2)$ \pillar operations. If $|A| \le 864 k$, we use the algorithm of \cref{lem:few} that outputs $\Oh(k^2)$ intervals of circular $k$-mismatch occurrences of $P$ in $\Oh(k^2)$ time in the \pillar model. Otherwise, we apply \cref{mt4} to $P_1$ and the part $T''$ of $T'$ spanned by the $k$-mismatch occurrences in $\Occ_k(P_1,T')$, obtaining a short primitive string $Q$ of length $q=|Q|$ such that $\Ham(P, Q ^\infty [0\dd |P|)) < 2k$.

\paragraph{Computing $V$}
The substring $V$ is computed using function \emph{ComputeV}.
Consider the occurrence of~$P_1$ at position $|P_2|$ of $P_2P_1P_2$.
Let us extend this substring to the right, trying to accumulate enough mismatches with a substring of $Q^\infty$ in order to reach the threshold specified in \cref{def:repreg}, which is $\Theta(k)$.
The subsequent mismatches are computed using \cref{lem:misoph} in $\Oh(k)$ total time in the \pillar model.
If we manage to accumulate enough mismatches, we call the resulting $k$-repetitive region~$R_{\rright}$.
We perform the same process by extending this occurrence of $P_1$ to the left, possibly obtaining a $k$-repetitive region $R_{\lleft}$.
Then, we let $V$ be the shortest substring $(P_2P_1P_2)[v \dd v')=P^2[\alpha_V \dd \beta_V]$ of~$P^2$ that spans both $R_{\lleft}$ (or $P_2P_1$ if $R_{\lleft}$ does not exist) and~$R_{\rright}$ (or $P_1P_2$ if $R_{\rright}$ does not exist).

\begin{center}
\fbox{
\begin{minipage}{13cm}
\vspace*{0.1cm}
\hspace*{-0.2cm}  {\bf function} \emph{ComputeV}

\smallskip
Compute $Q$ (\cref{mt4});

Let $Z=W= (P_2P_1P_2)[|P_2| \dd m)$;

\smallskip
Use \texttt{MismGenerator} to extend $W$ to the right until at least one 
of the following two conditions is satisfied: 
\begin{enumerate}[(a)]
    \item
we reach the end of $P_2P_1P_2$;
    \item
$W$ is a repetitive region w.r.t.~$Q$.
\end{enumerate}

If $W \ne P_1P_2$ then $R_{\rright}:=W$;

\smallskip
Use \texttt{MismGenerator}$^R$ to extend $Z$ to the left until any of \\
the following two conditions is satisfied:
\begin{enumerate}[(a)]
\item we reach the beginning of $P_2P_1P_2$; 
\item $Z^R$ is a repetitive region w.r.t.~the rotation of $Q^R$ that is a prefix of $P_1^R$.
\end{enumerate}

If $Z \ne P_2P_1$ then $R_{\lleft}:=Z$;

$V:=$\,the shortest substring of $P_2P_1P_2$ containing $Z$ and $W$;

\smallskip
{\bf return} $V$, $R_{\lleft}$, $R_{\rright}$;
\vspace*{0.1cm}
\end{minipage}
}
\end{center}

Let us observe that $V$ is at distance at most $2\cdot \lceil 8km/m \rceil =16k$ from a prefix of $(\rot_{r(P)}(Q))^\infty$, where $r(P) = (v - |P_2|) \bmod{q}$;
this follows by the definition of $k$-repetitive regions and the fact that $|R_{\rright}|,|R_{\lleft}|\leq m$. Moreover, obviously, $|V| \le 2m$. (Actually, $|V| \le \ceil{\frac32m}$.)

The rotations of $P$ that contain $P_1$ are in one-to-one correspondence with the length-$m$ substrings of $P_2P_1P_2$.
Each such substring contains $R_{\lleft}$, contains $R_{\rright}$, or is contained in $V$.
We now show that we can efficiently compute circular $k$-mismatch occurrences of $P$ that imply $k$-mismatch occurrences of either $R_{\lleft}$ or $R_{\rright}$ (if they exist) using the tools that were developed in \cref{subsec:few}.
We focus on $R_{\rright}$ as $R_{\lleft}$ can be handled symmetrically.
If $T'=T[a \dd b]$, then denote $\hat{T}:=T[a \dd \min(n-1,b+|R_{\rright}|-|P_1|)]$.
Due to \cref{repreg}, $R_{\rright}$ has $\cO(k)$ $k$-mismatch occurrences in~$\hat{T}$,
and they can be found in $\Oh(k^2 \log \log k)$ time plus $\Oh(k^2)$ \pillar operations using \cref{mt8}.
In the end, we compute $\ancycoc_k(T,\hat{T},P,R_{\rright})$ using \cref{lem:few} in $\Oh(k^2)$ time in the \pillar model. The output is a union of $\Oh(k^2)$ intervals.

We are left with computing elements of $\ancycoc_k(T,T',P,P_1)$ corresponding to $k$-mismatch occurrences of length-$m$ substrings of $V$ in $T$ in the case where $|V|\geq m$.

\paragraph{Computing $U$}
The substring $U$ is computed using function \emph{ComputeU}.
Let us take the considered \fragment $T''=T[y \dd y')$ in $T$, which by \cref{mt4} is at distance at most $6k$ from a prefix of~$Q^\infty$, and extend it to the right until either of the following three conditions is satisfied: (a) we reach the end of~$T$; (b) we have appended $\ceil{m/2}$ letters; or (c) the resulting substring has $18k$ additional mismatches with the same-length prefix of $Q^\infty$.
Symmetrically, we extend $T''$ to the left.
We set the obtained substring $T[u \dd u')$ of $T$ to be $U$.
We observe that $|U|\le 2m$ (since $\floor{3m/4}\ge |T'| \ge |T''| > 864k$ implies $m > 1152$ and $|U|\leq |T''|+2 \ceil{m/2} \leq 7m/4+2$). By \cref{mt4}, $U$ is at distance at most $6k+2\cdot 18k = 42k$ from a prefix of $(\rot_{r(T)}(Q))^\infty$, where $r(T) = (u - a) \bmod{q}$. If $|U|<m$, we do not construct the instance of \textsc{PeriodicSubMatch}.

\begin{center}
\fbox{
\begin{minipage}{13cm}
\vspace*{0.1cm}
\hspace*{-0.2cm}  {\bf function} \emph{ComputeU}

\smallskip
Initially $U=T''$;

\smallskip
Use \texttt{MismGenerator} to extend $U$ to the right until at least one of the following three conditions is satisfied: 
\begin{enumerate}[(a)]
    \item we reach the end of $T$; 
    \item we have appended $|P_2|$ letters; or 
    \item the appended \fragment of $T$ has $18k$ additional mismatches with the same-length prefix of $Q^\infty$;
\end{enumerate}

\smallskip
Use \texttt{MismGenerator}$^R$ to extend $U$ to the left until any of the following three conditions is satisfied: 
\begin{enumerate}[(a)]
    \item we reach the beginning of $T$; 
    \item we have prepended $|P_2|$ letters; or 
    \item the prepended \fragment of $T$ has $18k$ additional mismatches with the same-length suffix of $Q^{|T|}$;
\end{enumerate}

\smallskip
{\bf return} $U$;
\vspace*{0.1cm}
\end{minipage}
}
\end{center}

In the call to \textsc{PeriodicSubMatch}, we set $Q := \rot_{r(T)}(Q)$ and
$r := (r(P) - r(T)) \bmod{q}$.
Now, since $Q$ is primitive, it does not match any of its non-trivial cyclic rotations (this follows by Fine and Wilf's periodicity lemma, \cite{fine1965uniqueness}).
By \cref{mt4}, we have at least $128k-42k-1$ (resp.~$128k-16k-1$) exact occurrences of $Q$ in any length-$m$ \fragment of $U$ (resp.~$V$), so the periodicities must be synchronized in any circular $k$-mismatch occurrence.
Hence, for any $p\in \ancycoc_k(T,T',P,P_1)$ that corresponds to
$U^{(p)} =_k V^{(x)}$
we must have $p + r(T)\equiv x + r(P)  \pmod{q}$,
and therefore $p - x \equiv r(P) - r(T) \equiv r \pmod{q}$.

\paragraph{Correctness}
It suffices to show that there is no position $p\in \ancycoc_k(T,T',P,P_1)$ such that~$T^{(p)}$ is at distance at most $k$ from a substring of $V$ yet $[p \dd p+m) \not\subseteq [u \dd u')$.
Suppose that this is the case towards a contradiction, and assume without loss of generality that $p<u$ (the other case is symmetric); cf. \cref{fig:xyxy}.
We notice that $p<u$ is possible only if we stopped extending to the left when computing $U$ because we accumulated enough mismatches.
Further, let the implied $k$-mismatch occurrence of $P_1$ start at some position $t$ of $T$,
let $x$ be an integer such that the considered occurrence of sample $P_1$ in $V$ is at position $x+(t-p)$ (i.e., $\Ham(V^{(x)}, T^{(p)})\le k$), and let $F$ be the length-$(t-p)$ suffix of $Q^{|T|}$, i.e., $F=Q^{|T|}[|T| \cdot |Q|-(t-p)\dd |T| \cdot |Q|)$.
Then, via the triangle inequality, we have
\begin{multline*}
k \geq \Ham(V^{(x)}, T^{(p)})
\geq \Ham(V^{(x)}[0\dd t-p),T[p\dd t))\\
\geq  \Ham(T[p\dd t), F) - \Ham(F, V^{(x)}[0\dd t-p))
\geq 18k -16k
> k;
\end{multline*}
see \cref{fig:xyxy}. We obtain a contradiction, thus completing the reduction.
\end{proof}

\begin{figure}[htpb]
\centering
\begin{tikzpicture}[xscale=0.65,yscale=0.42]
    \draw[thick,xshift=1.5cm] (-3,0) rectangle node {$T$} (12,1);
    \draw[thick,xshift=1.5cm] (0.2,1) rectangle node {$U$} (10.5,2);
    \draw[thick] (3.5,2) rectangle node {$T''$} (10,3);
    \begin{scope}
    \draw (1.7,3) rectangle (3.5,4);
    \foreach \x in {2,2.5,3}{\draw (\x,3) -- (\x,4);}
    \draw (3.25,4) node[above] {$Q$};
    \draw[thick] (3,3) rectangle (3.5,4);
    \foreach \x in {1.8,2,2.3,2.4,2.7,2.9,3.2,3.3}{\draw[thick,blue] (\x,2) -- (\x,3);}
    \end{scope}
    \draw (1.7,2.5) node[left] {\textcolor{violet}{$\ge 18k$ mismatches}};
    \draw[ultra thick,orange] (1,0) rectangle (4,1);

    \begin{scope}[yshift=-7cm]
    \draw[thick,brown!70!black] (-1,0) rectangle (14,1);
    \draw[thick,green!20!black] (4,0) -- (4,1)  (9,0) -- (9,1);
    \draw[] (1.5,0.5) node {$P_2$};
    \draw[] (6.5,0.5) node {$P_1$};
    \draw[] (11.5,0.5) node {$P_2$};
    \draw[thick] (0.3,1) rectangle node {$V$} (11.5,2);
    \draw[densely dotted] (1,1) -- (1,8);
    \draw (1,7.5) node[right] {$p$};
    \draw (1,1.5) node[right] {$x$};
    \draw[densely dotted] (4,1) -- (4,8);
    \draw (4,7.5) node[right] {$t$};

    \begin{scope}
    \draw (0.3,3) rectangle (4,4);
    \foreach \x in {0.5,1,...,3.5}{\draw (\x,3) -- (\x,4);}
    \draw (3.75,4) node[above] {$Q$};
    \draw[thick] (3.5,3) rectangle (4,4);
    \foreach \x in {0.4,1.3,1.5,1.6,2.2,2.9,3.7}{\draw[thick,blue] (\x,2) -- (\x,3);}
    \end{scope}
    \draw (0.3,2.5) node[left] {\textcolor{violet}{$\le 16k$ mismatches}};
    \draw[ultra thick,green!70!black] (1,1) rectangle (4,2);
    \end{scope}
\end{tikzpicture}
\caption{Illustration of the correctness proof in \cref{lem:toPSM}. The substrings $T[p\dd t)$ and $V^{(x)}[0\dd t-p)$ denoted by orange and green rectangles, respectively, cannot match with up to $k$ mismatches, because they have at least $18k$ and at most $16k$ misperiods, respectively. Misperiods (that is, mismatches with $Q^\infty$) are drawn schematically as red vertical lines.}\label{fig:xyxy}
\end{figure}

\subsection{The Reporting Version of \texorpdfstring{$k$-Mismatch}{k-Mismatch} CPM}
In this section, we give a solution to \textsc{PeriodicSubMatch}.
Let us recall the notion of misperiods that was introduced in~\cite{DBLP:journals/jcss/Charalampopoulos21}.

\begin{definition}\label{def:misper}
A position $a$ in $S$ is a \emph{misperiod} with respect to a substring $Q$ of $S$ if $S[a] \ne Q^\infty[a]$. We denote the set of misperiods by $\Misp(S,Q)$.
\end{definition}

With \cref{lem:misoph}, we can compute the sets $I=\Misp(U,Q)$ and $J=\Misp(V,\rot_r(Q))$ in $\Oh(k)$ time in the \pillar model.
We define
\[\Mispers(i,j)\,=\,|I \cap [i \dd i+m)|+|J \cap [j \dd j+m)|.\]

The following problem is a simpler version of \textsc{PeriodicSubMatch} that was considered in~\cite{DBLP:journals/jcss/Charalampopoulos21}.\footnote{Actually,~\cite{DBLP:journals/jcss/Charalampopoulos21} considered a slightly more restricted problem which required that no two misperiods in $U^{(p)}$ and $V^{(x)}$ are aligned and computed a superset of its solution that corresponds exactly to the statement below.}

\defproblem{\textsc{PeriodicPeriodicMatch}$(U,V)$}{Same as in \textsc{PeriodicSubMatch}.}{
The set of $p\in [0\dd |U|-m]$  for which there exists a
position $x\in [0\dd |V|-m]$ in $V$ such that $p-x \equiv r \pmod{q}$ and $\Mispers(p,x)\le k$ (and thus $U^{(p)} =_k V^{(x)}$).
}

\smallskip
For an integer set $A$ and an integer $r$, let $A\oplus r = \{a+r\::\; a\in A\}$. An \emph{interval chain} for an interval~$I$ and non-negative integers $a$ and $q$ is a set \[\Chain_q(I,a):=I\,\cup\, (I\oplus q)\,\cup\, (I\oplus 2q)\,\cup \dots\cup\, (I\oplus aq).\]
Here, $q$ is called the \emph{difference} of the interval chain.

\begin{lemma}[{\cite[Lemma 15]{DBLP:journals/jcss/Charalampopoulos21}}]\label{lem:ppm}
Given sets $I=\Misp(U,Q)$ and $J=\Misp(V,\rot_r(Q))$, a solution to \textsc{PeriodicPeriodicMatch}, represented as $\Oh(k^2)$ interval chains each with difference $q$, can be computed in $\Oh(k^2)$ time.
\end{lemma}

Next, we observe that if $U^{(p)}=_kV^{(x)}$, then either some two misperiods in $U^{(p)}$ and $V^{(x)}$ are aligned, or the total number of misperiods in these substrings is at most $k$. We recall that $V=P^2[\alpha_V \dd \beta_V]$ to obtain the next observation.

\begin{observation}\label{obs:toPPM}
The result of $\textsc{PeriodicSubMatch}(U,V)$ can be computed as
\[{\textsc{PeriodicPeriodicMatch}(U,V)
\;\cup\; \bigcup_{i \in I,\,j \in J}\, \textsc{PairMatch}_k(U, i, P, (\alpha_V+j) \bmod m).}\]
\end{observation}

By the observation, there are $\Oh(k^2)$ instances of \textsc{PairMatch}$_k$, each taking $\Oh(k)$ time in the \pillar model, and a \textsc{PeriodicPeriodicMatch} instance that can be solved in $\Oh(k^2)$ time. This results in total time complexity $\Oh(k^3)$ for \textsc{PeriodicSubMatch}.
To complete the algorithm for the reporting version of $k$-Mismatch CPM in the \pillar model, it suffices to transform the output from a union of intervals and interval chains to a list of circular $k$-mismatch occurrences. This is performed via the following auxiliary problem.

\defproblem{\textsc{RectangleFamiliesUnions}}{Families $\F_1,\ldots,\F_f$ consisting of rectangles in $[0 \dd N)^2$ with $r:=\sum_{i=1}^f |\F_i|$ and $t := \max_i |\F_i|$.}{Sets $\mathit{Out}_i:=\bigcup \F_i$ of grid points of total size $s:=\sum_i\,|\mathit{Out}_i|$, with the points in each set reported in lexicographical order, that is, row by row from top to bottom and left-to-right in each row.}

We summarize the above transformations and perform the reduction.

\begin{lemma}\label{lem:ham_rep_red}
The reporting version of $k$-Mismatch CPM can be reduced to \textsc{RectangleFamiliesUnions} problem with $f=\Oh(n/m)$, $s\le |\cyc{}{k}(P,T)|$, $t=\Oh(k^3)$, and $N=n$ using $\Oh((n/m) \cdot k^3)$ \pillar operations and $\Oh((n/m) \cdot k^3+|\cyc{}{k}(P,T)|)$ extra time.
\end{lemma}
\begin{proof}
We apply \cref{lem:toPSM} which, using $\Oh((n/m)k^2\log \log k)$ time plus $\Oh((n/m)k^2)$ \pillar operations, outputs $\Oh((n/m)k^2)$ intervals of circular $k$-mismatch occurrences, which we directly report, and also returns $\cO(n/m)$ instances of \textsc{PeriodicSubMatch}.
Each of the instances can be reduced, using \cref{obs:toPPM}, to an instance of \textsc{PeriodicPeriodicMatch} and $\Oh(k^2)$ instances of \textsc{PairMatch}$_k$. The former can be solved with \cref{lem:ppm} in $\Oh(k^2)$ time with the output represented as a union of $\Oh(k^2)$ interval chains with difference $q$, while the latter can be solved with \cref{lem:PairMatch} in $\Oh(k^3)$ time in the \pillar model and produce $\Oh(k^3)$ intervals.
Overall, the computations for each instance take $\Oh(k^3)$ time in the \pillar model and produce $\Oh(k^3)$ intervals and $\Oh(k^2)$ interval chains with difference $q$.

\newcommand{\G}{\mathcal{G}}
For integer $q \ge 1$, by $\G_q$ we denote a grid of width $q$ obtained by fitting the integers from $[0 \dd n)$ so that the first row consists of numbers 0 through $q-1$, the second of numbers $q$ to $2q-1$, etc. We use the following convenient representation of interval chains.

\begin{claim}[{\cite[Lemma 5.4]{DBLP:journals/iandc/KociumakaRRSWZ22},\cite[Lemma 4]{DBLP:conf/esa/Charalampopoulos20a}}]\label{lem:ic_rect}
The set $\Chain_q(I, a)$ is a union of $\Oh(1)$ axes-aligned orthogonal rectangles in $\G_q$.
The coordinates of the rectangles can be computed in $\Oh(1)$ time.
\end{claim}

For each of the $\cO(n/m)$ instances of \textsc{PeriodicSubMatch}, we construct a grid $\G_q$ and insert $\cO(1)$ rectangles obtained from the corresponding instance of \textsc{PeriodicPeriodicMatch} and $\cO(1)$ rectangles for each of the $\cO(k^3)$ intervals returned by the $\cO(k^2)$ calls to \textsc{PairMatch}$_k$.\footnote{An interval is a special case of an interval chain (of any difference).}

Both dimensions of each constructed grid $\G_q$ are at most $n$.
Thus, all instances of \textsc{PeriodicSubMatch} reduce to \textsc{RectangleFamiliesUnions} in $\G = [0 \dd n)^2$ with $r=\Oh((n/m)\cdot k^3)$ and $t=\Oh(k^3)$.
For each instance of \textsc{PeriodicSubMatch}, the order in which the points are reported is identical to the order of the corresponding positions in text $T$.
For each such point $(i,j)$, we report $\alpha_U + q \cdot i + j \in \cyc{}{k}(P,T)$.

Due to \cref{lem:toPSM}, each element of $\cyc{}{k}(P,T)$ is reported a constant number of times.
\end{proof}

\newcommand{\Tsort}{\T_{\mathit{sort}}}
Let $\Tsort(n, U)$ represent the time complexity of sorting $n$ integers from $[0\dd U)$.

\begin{lemma}\label{lem:KleeReportPILLAR}
\textsc{RectangleFamiliesUnions} can be solved in $\cO(\sum_{i=1}^f \Tsort(|\F_i|,N)+s)$ time or in $\cO(\Tsort(r,N)+s)$ time.
\end{lemma}
\begin{proof}
The \textsc{RectangleFamiliesUnions} problem can be solved in $\Oh(N+r+s)$ time using a line-sweeping algorithm \cite[Claim 20]{DBLP:conf/esa/Charalampopoulos20a}.
The algorithm uses $\Oh(N+r)$ time to sort events in the sweep (corresponding to coordinates of rectangles' vertices) using radix sort and the remaining operations take $\Oh(s)$ time.
We use the same algorithm but replace radix sort with an arbitrary sorting algorithm, either for each family $\F_i$ separately or for all of them at once, thus obtaining the stated time complexity.
\end{proof}

\begin{theorem}\label{thm:ham_rep}
The reporting version of $k$-Mismatch CPM can be solved using\linebreak $\Oh((n/m) \cdot \Tsort(\Oh(k^3),\Oh(n))+\Output)$ time or $\Oh(\Tsort(\Oh((n/m) \cdot k^3),\Oh(n))+\Output)$ time, plus $\Oh((n/m) \cdot k^3)$ \pillar operations.
\end{theorem}
\begin{proof}
\cref{lem:ham_rep_red,lem:KleeReportPILLAR} imply our solution to $k$-Mismatch CPM in the \pillar model. 
\end{proof}

With integer sorting~\cite{DBLP:journals/jal/Han04}, in which $\Tsort(n,U)=\Oh(n \log \log n)$, we obtain the following corollary.

\begin{corollary}\label{cor:ham_rep}
The reporting version of $k$-Mismatch CPM can be solved using $\Oh((n/m) \cdot k^3 \log\log k+\Output)$ time plus $\Oh((n/m) \cdot k^3)$ \pillar operations.
\end{corollary}

One could also use randomized integer sorting~\cite{DBLP:conf/focs/HanT02} to obtain $\Oh((n/m) \cdot k^3 \sqrt{\log\log k}+\Output)$ expected time in the corollary.

\subsection{A Faster Algorithm for the Decision Version of \texorpdfstring{$k$-Mismatch}{k-Mismatch} CPM}

Two aligned misperiods can correspond to zero or one mismatch, while each misaligned misperiod always yields one mismatch; cf.\ \cref{fig:setting}.
Let us recall that $I=\Misp(U,Q)$ and $J=\Misp(V,\rot_r(Q))$.
We define the following \emph{mismatch correcting} function that can be used to correct \emph{surplus mismatches}:
\[\nabla(i,j)\, =\,  \begin{cases}
  0  &  \text{if}\ \ (i,j) \notin I\times J,\text{ otherwise:} \\
  1 & \text{if}\ \ U[i]\,\neq\, V[j],\\
  2 & \text{if}\ \ U[i]\,=\, V[j].
\end{cases}\]
Further, let 
$\Surplus(i,j)\,=\,\sum_{t=0}^{m-1} \nabla(i+t,j+t)$.

\defproblem{Decision \textsc{PeriodicSubMatch}$(U,V)$}
{Same as before, with the sets $I$, $J$ stored explicitly.}
{Any position $p \in [0 \dd |U|-m]$ such that 
$\Mispers(p,x)-\Surplus(p,x)\le k$ for some $x \in [0 \dd |V|-m]$ such that $p-x \equiv r \pmod{q}$, if there exists one.}

We consider a $(|U|-m+1) \times (|V|-m+1)$ grid $\mathcal{G}$.
A point $(i,j)\in \mathcal{G}$ is called \emph{essential} if $i\in I$, $j\in J$, 
and $i-j \equiv r \pmod{q}$.
The $\delta$-th diagonal in $\mathcal{G}$ consists of points $(i,j)$ that satisfy $i-j=\delta$; it is an \emph{essential} diagonal if it contains an essential point.
Let us observe that only essential points influence the
function $\Surplus$, and the number of these points is $\Oh(k^2)$.
A combination of this observation with a simple 1D sweeping 
algorithm implies the following lemma.

\begin{lemma}[Compact representation of $\Mispers$ and $\Surplus$]\label{WR}
In $\Oh(k^2 \log \log k)$ time, we can:
\begin{enumerate}[(a)]
\item\label{ia} Partition the grid $\mathcal{G}$ by $\Oh(k)$ vertical and $\Oh(k)$ horizontal lines into $\Oh(k^2)$ disjoint rectangles such that
the value $\Mispers(i,j)$ is the same for all  points $(i,j)$ in a single rectangle. Each rectangle stores the value $\Mispers(i,j)$ common to all points $(i,j)$ that it contains.
\item\label{ib} Partition all essential diagonals in $\mathcal{G}$ into $\Oh(k^2)$ pairwise disjoint diagonal segments such that the 
value $\Surplus(i,j)$ is the same for all points $(i,j)$ in a single segment.
Each segment stores the value $\Surplus(i,j)$ common to all points $(i,j)$ that it contains.
\end{enumerate}
\end{lemma}
\begin{proof}
Partitioning~\eqref{ia}: We partition the first axis (second axis) into $\Oh(k)$ axis segments such that, for all $i$ ($j$, respectively) in the same segment, the set $I \cap [i \dd i+m)$ ($J \cap [j \dd j+m)$, respectively) is the same. Then, we create rectangles being Cartesian products of the segments.

The partitioning of each axis is performed with a 1D sweep; we describe it in the context of the first axis.
For each $i \in I$, we create an event at position $i-m+1$, where the misperiod $i$ is inserted, and an event at position $i+1$ (if $i+1 \le |U|-m+1$), where it is removed. We can now sort all events in $\Oh(k \log \log k)$ time using integer sorting~\cite{DBLP:journals/jal/Han04} and process them in order, storing the number of active misperiods. For all $i$ in a segment without events, the set $I \cap [i \dd i+m)$ is the same.

We obtain $\Oh(k)$ segments on each axis, yielding $\Oh(k^2)$ rectangles. Part~\eqref{ia} takes $\Oh(k^2)$ time.

Partitioning~\eqref{ib}:
First, we group essential points by (essential) diagonals using integer sorting. On each essential diagonal, we sort the essential points bottom-up and perform the same kind of 1D sweep as in~\eqref{ia}, using $\nabla$ to compute the weights of the events. The whole algorithm works in $\Oh(k^2 \log \log k)$ time.
\end{proof}

\newcommand{\val}{\mathsf{val}}
\newcommand{\cell}{\mathsf{cell}}
\newcommand{\diag}{\mathsf{diag}}
We assume that the grid $\mathcal{G}$ is partitioned by selected horizontal and vertical
lines into disjoint rectangles, called \emph{cells}. These cells and some diagonal segments are weighted.
Let us denote by
$\cell(i,j)$ and $\diag(i,j)$  the weight of the cell and the diagonal
segment, respectively, containing point~$(i,j)$. In the following problem, we only care about points on diagonal segments.

\defproblem{\textsc{DiagonalSegments}}{
A grid partitioned by $\Oh(k)$ vertical and $\Oh(k)$ horizontal lines into $\Oh(k^2)$ weighted rectangles, called cells, and $\Oh(k^2)$ pairwise disjoint weighted diagonal line segments, all parallel to the line that passes through $(0,0)$ and $(1,1)$.}{
Report a point $(x,y) \in \mathbb{Z}^2$ on some diagonal line segment with minimum value
$\val(x,y) := \cell(x,y) + \diag(x,y)$.
}

\begin{figure}[ht]
    \centering
    \begin{tikzpicture}[scale=0.65]
        \tikzstyle{dot}=[inner sep=2pt, circle, draw, fill=gray]
        \draw[thick,densely dotted] (2,0) -- (2,7) (7,0) -- (7,7) (0,2) -- (12,2) (0,5) -- (12,5);
        \draw[very thick] (2,2) rectangle (7,5);
        \draw[violet,thick] (2.5,3) -- (4,4.5);
        \draw[brown,thick] (5,4) -- (7.2,6.2);
        \draw[red,thick] (3,1) -- (5,3);
        \draw[blue,thick] (5.5,3.5) -- (6.5,4.5);
        \draw[green!70!black,thick] (5,1.5) -- (9,5.5);
        \foreach \x/\y in {2.5/3,4/4.5,5/4,6/5,4/2,5/3,5.5/3.5,6.5/4.5,5.5/2,7/3.5}{
            \filldraw (\x,\y) circle (0.11cm);
        }
    \end{tikzpicture}
    \caption{The weight of the distinguished cell is equal to $\Mispers(x,y)$ for all points $(x,y)$ in that cell.
The weight of a single diagonal segment
is equal to $-\Surplus(x,y)$ for all points $(x,y)$ that lie on that segment.
In the \textsc{DiagonalSegments} problem, we are to find any point $(x,y)$ on some diagonal segment that minimizes the sum of the weight of its cell and of its diagonal segment, i.e., $\Mispers(x,y)-\Surplus(x,y)$.
To this end, it suffices to consider endpoints of diagonal segments and crossings of diagonal segments with rectangles' boundaries (all shown as dots).}\label{fig:cells_diags}
\end{figure}

\cref{fig:cells_diags} provides intuition for our solution to the \textsc{DiagonalSegments} problem.

\begin{lemma}\label{lem:geo}
The \textsc{DiagonalSegments} problem can be solved in $\Oh(k^2 \log k / \log \log k)$ time.
\end{lemma}
\begin{proof}
A sought point $(x,y)$ that minimizes $\val(x,y)$ either (1) is an endpoint of a diagonal segment or (2) lies on the intersection of a diagonal segment and an edge of a cell.

In case (1), it suffices to identify, for all endpoints of all diagonal segments, the cells to which they belong. Let us assume that vertical and horizontal lines partition the grid into columns and rows, respectively. Then, each cell can be uniquely identified by its column and row. By the following claim, we can compute, in $\Oh(k^2 \log \log k)$ time, for all queried points, the rows they belong to; the computation of columns is symmetric.

\begin{claim}\label{lem:nearest_line}
Given $\Oh(p)$ horizontal lines and $\Oh(p)$ points, one can compute, for each point, the nearest line above it in $\Oh(p \log \log p)$ time in total.
\end{claim}
\begin{proof}
We create a list of integers containing the vertical coordinates of all queried points and all lines. Then, we sort the list in $\Oh(p \log \log p)$ time. The required answers can be retrieved by a simple traversal of the sorted list.
\end{proof}

In case (2), let us consider intersections with horizontal edges; the intersections with vertical edges can be handled symmetrically. Assume $x$ is the horizontal coordinate. We perform an affine transformation of the plane $(x,y) \mapsto (y-x,y)$ after which diagonal line segments become vertical, but horizontal line segments remain horizontal. The sought points can be computed using the following claim in $\Oh(k^2 \log k / \log \log k)$ time.

\begin{claim}\label{VH-seg}
Given $s$ vertical and horizontal weighted line segments such that no two line segments of the same direction intersect, an intersection point of a vertical and a horizontal line segment with minimum total weight can be computed in $\Oh(s \log s / \log \log s)$ time.
\end{claim}
\begin{proof}
We perform a left-to-right line sweep. The events are vertical line segments as well as the beginnings and endings of horizontal line segments; they can be sorted by their $x$-coordinates in $\Oh(s \log \log s)$ time using integer sorting.
The horizontal line segments intersecting the sweep line are stored using a dynamic predecessor data structure~\cite{DBLP:journals/ipl/Willard83}, and their weights in the same order are stored in a dynamic RMQ data structure of size $\Oh(s)$ that supports insertions, deletions, and range minimum queries in amortized $\Oh(\log s / \log \log s)$ time~\cite{DBLP:conf/wads/BrodalDR11}.
This way, when considering a vertical line segment, we can compute the minimum-weight horizontal line segment that intersects it in $\Oh(\log s / \log \log s)$ time. 
\end{proof}
This concludes our solution to \textsc{DiagonalSegments}.
\end{proof}

\begin{theorem}\label{thm:ham_dec}
The decision version of the $k$-Mismatch CPM problem can be solved in $\Oh((n/m) \cdot k^2 \log k / \log\log k)$ time plus $\Oh((n/m) \cdot k^2)$ \pillar operations.
\end{theorem}
\begin{proof}
We apply \cref{lem:toPSM}, which, using $\Oh((n/m)k^2\log \log k)$ time plus $\Oh((n/m)k^2)$ \pillar operations, outputs $\Oh((n/m)k^2)$ intervals of circular $k$-mismatch occurrences and returns $\Oh(n/m)$ instances of \textsc{PeriodicSubMatch}.
Next, we use the geometric interpretation of \textsc{PeriodicSubMatch}.
The weight of a cell is the value $\Mispers(i,j)$ common to
all points $(i,j)$ in this cell. Similarly, the weight of a diagonal 
segment equals $-\Surplus(i,j)$, the number of surplus misperiods 
that we have to subtract for all points in this segment.
These values are computed in $\Oh(k^2 \log \log k)$ time using \cref{WR}.
Now, the decision version of the \textsc{PeriodicSubMatch} problem is reduced to finding a point minimizing
the value $\Mispers(i,j)-\Surplus(i,j)$, which is the sum of weights of a cell and a diagonal segment meeting at the point $(i,j)$ if $\Surplus(i,j) \neq 0$.
The decision version of \textsc{PeriodicSubMatch} is thus reduced, in $\Oh(k^2\log\log k)$ time plus $\Oh(k^2)$ \pillar operations, to
one instance of each of the \textsc{DiagonalSegments} problem (if the sought point is on a diagonal segment) and the \textsc{PeriodicPeriodicMatch} problem (otherwise). 
The thesis then follows from \cref{lem:geo} and \cref{lem:ppm}.
\end{proof}

\begin{remark}
Using this geometric approach, the reporting version of $k$-Mismatch CPM can also be solved in $\Oh((n/m) \cdot k^2 \log^{\Oh(1)} k + k \cdot \Output)$ time
plus $\Oh((n/m)\cdot k^2)$ \pillar operations.
\end{remark}

\section{Fast $k$-Mismatch CPM in Important Settings}\label{app:settings}
Recall that, in the \pillar model, we measure the running time of an algorithm with respect to the total number of a few primitive operations, which are executed on a collection $\mathcal{X}$ of input strings. Then, for any fixed setting, an efficient implementation of these primitive operations yields an algorithm for this setting.
In this section, we combine \cref{thm:ham_dec,thm:ham_rep} with known implementations of the primitive \pillar operations in the standard, internal, packed, and read-only settings (\cref{sec:std}), in the dynamic setting (\cref{sec:dyn}), in the fully compressed setting (\cref{sec:comp}), and in the quantum setting (\cref{sec:quantum}), thus obtaining efficient algorithms for $k$-Mismatch CPM in these settings.

\subsection{Standard, Internal, Packed, and Read-Only Settings}\label{sec:std}

Let $X$ be any string in collection $\mathcal{X}$.
We can access any substring $X[\ell\dd r]$ of $X$ using a pointer to~$X$ along with the two corresponding indices $\ell$ and $r$. As shown in~\cite{DBLP:conf/focs/Charalampopoulos20}, implementing the \pillar model in the standard setting is done by putting together a few well-known results: 
\begin{itemize}
    \item $\Extract$, $\Access$, and $\Length$ have trivial implementations.
    \item $\LCP$ queries are implemented by constructing the generalized suffix tree for $\mathcal{X}$ in linear time~\cite{DBLP:conf/focs/Farach97} and preprocessing it for answering lowest common ancestor queries in $\Oh(1)$ time~\cite{DBLP:conf/latin/BenderF02,DBLP:journals/siamcomp/HarelT84}. The same construction for the reversed strings in $\mathcal{X}$ implements $\LCP_R$ queries.
    \item $\IPM$ queries can be performed in $\Oh(1)$ time after a linear-time preprocessing using a data structure by Kociumaka et al.~\cite{DBLP:journals/siamcomp/KociumakaRRW24}.
\end{itemize}
The above discussion is summarized in the following statement.

\begin{theorem}\label{the:standardPILLAR}
After an $\Oh(n)$-time preprocessing of a collection of strings of total length $n$, each \pillar operation can be performed in $O(1)$ time.
\end{theorem}

For the reporting version of $k$-Mismatch CPM, we combine the algorithm stated in the \pillar model in \cref{thm:ham_rep} with \cref{the:standardPILLAR} and radix sort, in which $\Tsort(n,U)=\Oh(n+U)$. (Observe that $\Output \le n$.)
For the decision version of $k$-Mismatch CPM, we combine the algorithm stated in the \pillar model in \cref{thm:ham_dec} with \cref{the:standardPILLAR}. Overall, we obtain \cref{thm:ham_standard} that we restate here for convenience.

\hamstand*

Using the implementation of the \pillar model in the standard setting, we get the following result for the \emph{internal setting}, where we want to preprocess a string $S$ and then query for the $k$-mismatch circular occurrences of a \fragment of $S$ in another \fragment of $S$.

\begin{theorem}[Internal Setting]\label{the:CPMinternal}
Given a string $S$ of length $n$ over an integer alphabet, after an $\cO(n)$-time preprocessing,
given any two \fragments $P$ and $T$ of $S$ and an integer threshold $k >0$,
we can solve the reporting version of $k$-Mismatch CPM in $\Oh((|T|/|P|)\cdot k^3 \log\log k+\Output)$ time and the decision version in $\Oh((|T|/|P|)\cdot k^2 \log k / \log\log k)$ time.
\end{theorem}

One can obtain a faster algorithm in the so-called \emph{packed setting}, which is most relevant when the alphabet size $\sigma$ is small.
In this setting, we assume that the input string is given in a packed representation, with each machine word representing $\cO(\log_{\sigma} n )$ letters.
We use \cref{thm:ham_rep} with radix sort, in which $\Tsort(n,U)=\Oh(n+\sqrt{U})$, and \cref{thm:ham_dec}. With an optimal implementation of the \pillar model for the packed setting~\cite{DBLP:conf/stoc/KempaK19,DBLP:journals/siamcomp/KociumakaRRW24}, one obtains the following result.

\begin{theorem}[Packed Setting, \cite{DBLP:journals/siamcomp/KociumakaRRW24}]\label{the:CPMpacked}
Given $\cO(n / \log_\sigma n)$-size packed representations of a text~$T$ of length $n$ and a pattern $P$ of length $m$, and an integer threshold $k>0$,
we can solve the reporting and decision versions of $k$-Mismatch CPM in time $\Oh(n / \log_\sigma n + (n/m) \cdot k^3+\Output)$ and $\Oh(n / \log_\sigma n + (n/m) \cdot k^2 \log k / \log\log k)$, respectively.
\end{theorem}

An analogous result was observed in~\cite{DBLP:journals/siamcomp/KociumakaRRW24} but with $\Oh(n / \log_\sigma n + (n/m) \cdot k^3 \log\log k+\Output)$ time complexity of the reporting version derived from the conference version of this work~\cite{DBLP:conf/esa/Charalampopoulos22}.

In the read-only setting, we assume that we have read-only access to the input strings and wish to bound the extra space usage of our algorithms. As observed in~\cite{DBLP:conf/cpm/BathieCS24}, by combining \cref{cor:ham_rep,thm:ham_dec} with an efficient implementation of the \pillar model for this setting~\cite{DBLP:conf/cpm/KosolobovS24,DBLP:conf/cpm/BathieCS24}, one obtains the following result.

\begin{theorem}[Read-only Setting, {\cite{DBLP:conf/cpm/BathieCS24}}]
    Suppose that we have read-only random access to a text~$T$ of~length $n$ and a pattern $P$ of~length $m$ over an integer alphabet.
    Given an integer threshold~$k$, for any integer $\tau = \cO(m/\log^2m)$, we can solve the reporting version of $k$-Mismatch CPM in $\tilde{\cO}(n+(n/m) \cdot k^3\tau+\Output)$ time\footnote{In this work, the $\tilde{\cO}(\star)$ notation hides factors polylogarithmic in the length of the input strings.} using $\tilde{\cO}(m/\tau+k^3)$ extra space and the decision version of $k$-Mismatch CPM in $\tilde{\cO}(n+(n/m) \cdot k^2\tau)$ time using $\tilde{\cO}(m/\tau+k^2)$ extra space.
    \end{theorem}

\subsection{Dynamic Setting}\label{sec:dyn}

Let $\mathcal{X}$ be a growing collection of non-empty strings; it is initially empty, and then undergoes updates by means of the following operations:
\begin{itemize}
    \item $\texttt{Makestring}(U)$: Insert a non-empty string $U$ to $\mathcal{X}$
    \item $\texttt{Concat}(U,V)$: Insert string $UV$ to $\mathcal{X}$, for $U,V\in \mathcal{X}$
        \item $\texttt{Split}(U,i)$: Insert $U[0\dd i)$ and $U[i\dd |U|)$ to $\mathcal{X}$, for $U\in\mathcal{X}$ and $i\in[0\dd |U|)$.
    \end{itemize}

Note that $\texttt{Concat}$ and $\texttt{Split}$ do not remove their arguments from $\mathcal{X}$.
By $N$ we denote an upper bound on the total length of all strings in $\mathcal{X}$ throughout all updates executed by an algorithm.
Gawrychowski et al.~\cite{DBLP:conf/soda/GawrychowskiKKL18} showed an efficient data structure for maintaining $\mathcal{X}$ subject to $\LCP$ queries.
The authors of~\cite{DBLP:conf/focs/Charalampopoulos20} showed that this data structure can be augmented to also support all remaining \pillar operations; see~\cite{Panos} for a more direct description of an efficient algorithm answering $\IPM$ queries in this setting and \cite{DBLP:conf/spire/DuysterK24} for a faster implementation of $\IPM$ queries.
The data structure is Las-Vegas randomized, that is, it always returns correct results, but the running times are guaranteed with high probability (w.h.p.), that is, probability $1-1/N^{\Omega(1)}$.

\begin{theorem}[\cite{DBLP:conf/soda/GawrychowskiKKL18,DBLP:conf/spire/DuysterK24}]\label{the:dynamicPILLAR}
A collection $\mathcal{X}$ of non-empty persistent strings of total length $N$ can be dynamically maintained with operations $\mathtt{Makestring}(U)$,  $\mathtt{Concat}(U,V)$, $\mathtt{Split}(U,i)$ requiring time $\Oh(\log N+|U|)$, $\Oh(\log N)$ and $\Oh(\log N)$, respectively, so that \pillar operations can be performed in time $\Oh(\log N)$.
All stated time complexities hold with probability $1-1/N^{\Omega(1)}$.
\end{theorem}

Kempa and Kociumaka~\cite[Section 8 in the arXiv version]{DBLP:journals/corr/abs-2201-01285} have presented an alternative deterministic implementation of dynamic strings, which supports operations $\texttt{Makestring}(U)$, $\texttt{Concat}(U,V)$, $\texttt{Split}(U,i)$ in $\Oh(|U|\log^{\Oh(1)}\log N)$, $\Oh(\log|UV|\log^{\Oh(1)}\log N)$, and $\Oh(\log|U|\log^{\Oh(1)}\log N)$ time, respectively, so that \pillar operations can be performed in time $\Oh(\log N \log^{\Oh(1)}\log N)$.

Combining \cref{cor:ham_rep,thm:ham_dec} with \cref{the:dynamicPILLAR} and~\cite{DBLP:journals/corr/abs-2201-01285}, we obtain the following result. Let us note that $o(\log k)$ factors from \cref{cor:ham_rep,thm:ham_dec} are dominated by $\Oh(\log N)$ factors from \cref{the:dynamicPILLAR} and~\cite{DBLP:journals/corr/abs-2201-01285}.

\begin{theorem}[Dynamic Setting]\label{the:CPMdynamic}
A collection $\mathcal{X}$ of non-empty persistent strings of total length $N$ can be dynamically maintained with operations $\mathtt{Makestring}(U)$, $\mathtt{Concat}(U,V)$, $\mathtt{Split}(U,i)$ requiring time $\Oh(\log N+|U|)$, $\Oh(\log N)$ and $\Oh(\log N)$, respectively, so that, given two strings $P,T\in \mathcal{X}$ and an integer threshold $k>0$,
we can solve the reporting and decision versions of $k$-Mismatch CPM in $\Oh((|T|/|P|)\cdot k^{3}\log N+\Output)$ time and $\Oh((|T|/|P|)\cdot k^2 \log N)$ time, respectively.
All stated time complexities hold with probability $1-1/N^{\Omega(1)}$.

Randomization can be avoided at the cost of a $\log^{\Oh(1)} \log N$ multiplicative factor in all the update times, with the reporting and decision versions of $k$-Mismatch CPM queries answered in $\Oh((|T|/|P|)\cdot k^{3}\log N\log^{\Oh(1)}\log N+\Output)$ time and $\Oh((|T|/|P|)\cdot k^2 \log N\log^{\Oh(1)}\log N)$ time, respectively.
\end{theorem}

\subsection{Fully Compressed Setting}\label{sec:comp}

In the fully compressed setting, we want to solve $k$-Mismatch CPM when both the text and the pattern are given as straight-line programs.

\newcommand{\gen}{\mathsf{val}}
\newcommand{\PT}{\mathit{ParseTree}}

A straight line grammar is a context-free grammar $G$ that consists of a set $\Sigma$ of terminals and a set $N_G = \{A_1,\dots,A_n\}$ of non-terminals such that each $A_i \in N_G$ is associated with a unique production rule
$A_i \rightarrow f_G(A_i) \in (\Sigma \cup \{A_j : j < i\})^*$. 
We can assume without loss of generality that the grammar is a straight-line program (in Chomsky normal form), that is, each production rule is of the form $A \rightarrow BC$ for some symbols $B$ and $C$.
Every symbol $A \in S_G:=N_G \cup\Sigma$ generates a unique string, which we denote by $\gen(A) \in \Sigma^*$. The string $\gen(A)$ can be obtained from $A$ by repeatedly replacing each non-terminal with its production. In addition, $A$ is associated with its parse tree $\PT(A)$ consisting of a root labeled with A, with two subtrees $\PT(B)$ and $\PT(C)$ attached (in this order) if $A$ is a non-terminal $A\rightarrow BC$, or without any subtrees if $A$ is a terminal.

Let us observe that, if we traverse the leaves of $\PT(A)$ from left to right, spelling out the corresponding non-terminals,
then we obtain $\gen(A)$. We say that $G$ generates $\gen(G) := \gen(A_n)$. The parse tree $\PT(G)$ of $G$
is then the parse tree of the starting symbol $A_n \in N_G$.

As also observed in~\cite{DBLP:conf/focs/Charalampopoulos20}, given an SLP $G$ of size $n$ that generates a string $S$ of length $N$, we can efficiently implement the \pillar operations through dynamic strings. Let us start with an empty collection $\mathcal{X}$ of dynamic strings. Using $\Oh(n)$ $\texttt{Makestring}(a)$ operations, for $a \in \Sigma$, and $\Oh(n)$ $\texttt{Concat}$ operations (one for each non-terminal of $G$), we can insert $S$ to $\mathcal{X}$ in $\Oh(n \log N)$ time w.h.p.
Then, we can perform each \pillar operation in $\Oh(\log N)$ time w.h.p., due to \cref{the:dynamicPILLAR}.
Next, we outline a deterministic implementation of \pillar operations in the fully compressed setting.

Following~\cite{DBLP:conf/focs/Charalampopoulos20}, the handle of a \fragment $S = X[\ell \dd r]$ consists of a pointer to the SLP $G \in \mathcal{X}$
generating $X$ along with the positions $\ell$ and $r$. This makes operations $\Extract$ and $\Length$ trivial. As argued
in~\cite{DBLP:conf/focs/Charalampopoulos20}, all remaining \pillar operations admit efficient implementations in the considered setting.
\begin{itemize}
    \item For $\Access$, we use the data structure of Bille et al.~\cite{DBLP:journals/siamcomp/BilleLRSSW15} that can be built in $\cO(n \log(N/n))$ time and supports $\Access$ queries in $\cO(\log N)$ time.
    \item For $\LCP$ and $\LCP_R$, we use the data structure of I~\cite{DBLP:conf/cpm/I17} that is based on the recompression technique, which is due to Jeż~\cite{DBLP:journals/talg/Jez15,DBLP:journals/jacm/Jez16}. It can be built in $\cO(n \log(N/n))$ time and supports the said queries in $\cO(\log N)$ time.
    \item We answer $\IPM$ queries in $\cO(\log N)$ time using the algorithm of~\cite{DBLP:conf/spire/DuysterK24} on top of the data structure of I~\cite{DBLP:conf/cpm/I17}.
\end{itemize}

The above discussion is summarized in the following statement.

\begin{theorem}[see~\cite{DBLP:journals/siamcomp/BilleLRSSW15,DBLP:conf/cpm/I17,DBLP:conf/spire/DuysterK24}]\label{thm:SLP PILLAR}
Given a collection of SLPs of total size $n$ generating
strings of total length $N$, each \pillar operation can be performed in $\Oh(\log N)$ time after an
$\Oh(n \log (N/n))$-time preprocessing.
\end{theorem}

If we applied \cref{cor:ham_rep,thm:ham_dec} directly in the fully compressed setting, we would obtain $\Omega(N/M)$ time, where $N$ and $M$ are the uncompressed lengths of the text and the pattern, respectively.
We now present a more efficient algorithm for $k$-Mismatch CPM
in this setting; it is an adaptation of an analogous procedure provided in~\cite[Section 7.2]{DBLP:conf/focs/Charalampopoulos20}
for (non-cyclic) pattern matching with mismatches. We use the deterministic implementation of
the \pillar model from \cref{thm:SLP PILLAR}.

\begin{theorem}\label{thm:CPM SLP}
Let $G_T$ denote a straight-line program of size $n$ generating a string $T$ of length~$N$, let $G_P$ denote a
straight-line program of size $m$ generating a string $P$ of length $M \leq N$, and let $k > 0$ denote an integer threshold.
We can solve the decision version of $k$-Mismatch CPM in time $\Oh(m \log N + n k^{2} \log N)$ and
the reporting version in $\Oh(m \log N + n k^{3}\log N+\Output)$ time.
\end{theorem}
\begin{proof}
Let $\mathcal{X} := \{G_T , G_P \}$. The overall structure of our algorithm is as follows. We
first preprocess the collection $\mathcal{X}$ in $\Oh((n+m) \log N)$ time according to \cref{thm:SLP PILLAR}.
Next, we traverse~$G_T$. If $M \le 1$, we first check for occurrences at every terminal of $G_T$ in total time $\cO(n)$. Then we compute, for every non-terminal $A$ of $G_T$, the circular $k$-mismatch occurrences of $P$ in $T$ that “cross”~$A$.
For each non-terminal $A \in N_{G_T}$ with production rule $A \rightarrow BC$, let
\begin{align*}
A_{\ell} & := \gen(B)[ \max\{0, |\gen(B)| - M + 1\} \dd |\gen(B)|), \text{ and}\\
A_r & := \gen(C)[ 0 \dd \min\{ M-1 , |\gen(C)|\}).
\end{align*}
The following claim is immediate.
\begin{claim}\label{claim:slp_decomp}
The circular $k$-mismatch occurrences in $\gen(A)$ can be partitioned to occurrences in $\gen(B)$, $\gen(C)$, and $A_{\ell}A_r$.
\end{claim}

If we simply want to decide whether $T$ contains a circular $k$-mismatch occurrence of $P$, we can use \cref{thm:SLP PILLAR,thm:ham_dec} to decide whether any $\gen(A)$ for $A\in N_{G_T}$ contains an occurrence in $\Oh(k^2\log N)$ time per a non-terminal since $A_{\ell}A_r$ is a \fragment of $\gen(G_T)$ of length at most $2(M-1)$. This way, we obtain the total complexity of $\Oh(m\log N +n k^2\log N)$.

By using \cref{cor:ham_rep,thm:SLP PILLAR}, we can find all the occurrences in $A_{\ell}A_{r}$ for each $A\in N_{G_T}$ in time $\Oh(k^{3}\log N+\Output)$ per non-terminal.
Let us consider $G_T$ as a topologically-sorted DAG $\mathcal{G}$ in which each symbol is a vertex and for each production $A \to BC$ we have arcs $(A,B)$ and $(A,C)$. We then create shortcuts in $\mathcal{G}$ as follows. For each vertex (symbol), we know if there is a circular occurrence crossing it if it is a non-terminal or inside it if it is a terminal. For each vertex of $\mathcal{G}$, we compute a Boolean flag saying if it has any descendant with a circular occurrence. Let $\mathcal{G}'$ be a DAG obtained from $\mathcal{G}$ by removing all vertices with this flag set to false. For each maximal path $u_1,\ldots,u_r$ in $\mathcal{G}'$ all of whose vertices apart from the bottommost one have a single child, we create a shortcut from $u_1$ to $u_r$ in $\mathcal{G}$.
To compute $\cyc{}{k}(P,T)$, we can then use a dynamic programming approach in $\PT(G_T)$ exploiting \cref{claim:slp_decomp}:
all $k$-mismatch circular occurrences can be reported in time proportional to their number by performing the traversal by using the shortcuts computed for $\mathcal{G}$ to skip \fragments of~$T$ that do not contain any such occurrences.
The shortcuts ensure that we traverse a constant number of edges and shortcuts per reported $k$-mismatch circular occurrence.
The total time required for the reporting variant is hence as claimed.
\end{proof}

\subsection{Quantum Setting}\label{sec:quantum}

In what follows, we assume that the input strings can be accessed in a quantum query model~\cite{AMB04,DBLP:journals/tcs/BuhrmanW02}.
We are interested in the time complexity of our quantum algorithms~\cite{BBCplus}.

\begin{observation}[{\cite[Observation 2.3]{DBLP:journals/talg/JinN24}}]
For any two strings $S,T$ of length at most~$n$,
$\LCP(S, T)$ or $\LCP_R(S, T)$ can be computed in $\tilde{\cO}(\sqrt{n})$ time in the quantum model w.h.p.
\end{observation}

\begin{fact}[Corollary of {\cite{DBLP:journals/jda/HariharanV03}}, cf.~{\cite[Observation 39]{stacs24}}]
For any strings $S$ and $T$ of length at most~$n$, with $|T| \le 2|S|$,
$\mathsf{IPM}(S, T)$ can be computed in $\tilde{\cO}(\sqrt{n})$ time in the quantum model w.h.p.
\end{fact}

All other \pillar operations trivially take
$\cO(1)$ quantum time. As a corollary, in the quantum setting, all \pillar operations can be implemented in $\tilde{\cO}(\sqrt{m})$ quantum time with no preprocessing, as we always deal with strings of length $\cO(m)$.

\begin{theorem}[Quantum Setting]\label{thm:quantum}
The reporting version of the $k$-Mismatch CPM problem can be solved in $\tilde{\cO}(k^3 \cdot n/\sqrt{m}+\Output)$ time in the quantum model w.h.p.
The decision version of the $k$-Mismatch CPM problem can be solved in $\tilde{\cO}(k^2 \cdot n/\sqrt{m})$ time in the quantum model w.h.p.
\end{theorem}

\section{\texorpdfstring{$k$-Edit}{k-Edit} CPM}\label{sec:edit}

In this section, we solve the reporting version of the $k$-Edit CPM problem
and reduce the decision version of the $k$-Edit CPM problem to two instances of the \allkep problem.
Our approach is based on the following observation. 
\begin{observation}\label{obs:ed_cl}
If $T[i\dd r)$ is a $k$-edit circular occurrence of a pattern $P$ in a text $T$,
then there exists a position $j\in [i\dd r]$ and a decomposition $P=P_1P_2$
such that $\ed(T[i\dd j),P_2)+\ed(T[j\dd r),P_1)\le k$.
\end{observation}

Our algorithms consider all $j\in [0\dd n]$ and, in each iteration, identify the occurrences satisfying \cref{obs:ed_cl} with a fixed $j$. 
These occurrences play an analogous role to the set $\textsc{PairMatch}(T,j,P,0)$ in $k$-Mismatch CPM. 
Thus, one might wish to think of $T[j]$ as the position aligned with $P[0]$ in these occurrences; this intuition, however, is imperfect because some alignments may delete $P[0]$.

As hinted in the introduction, our approach relies on the $\LPref_{k'}[0\dd n]$ arrays for $k'\in [0\dd k]$,
where $\LPref_{k'}[j]$ is the maximum length of a prefix of $P$ at edit distance at most $k'$ from a prefix $T[j\dd n)$.
In our setting, once we fix a decomposition $P=P_1P_2$ and a \fragment $T[i\dd j)$ such that $\ed(T[i\dd j),P_2)=d\le k$,
the existence of a \fragment $T[j\dd r)$ such that $\ed(T[j\dd r),P_1)\le k-d$ is equivalent to $\LPref_{k-d}[j]\ge |P_1|$.
In the decision version of the $k$-Edit CPM problem, we also use symmetric $\LSuf_{k'}[0\dd n]$ arrays,
where $\LSuf_{k'}[j]$ is the maximum length of a suffix of $P$ at edit distance at most $k'$ from a suffix of $T[0 \dd j)$. 
Now, once we fix a decomposition $P=P_1P_2$ and a \fragment $T[j\dd r)$ such that $\ed(T[j\dd r),P_1)=d\le k$,
the existence of a \fragment $T[i\dd j)$ such that $\ed(T[i\dd j),P_2)\le k-d$ is equivalent to $\LSuf_{k-d}[j]\ge |P_2|$.
Combining \cref{obs:ed_cl} with the above insight on $\LPref_{k'}[0\dd n]$ and $\LSuf_{k'}[0\dd n]$ arrays, we derive the following observation.

\begin{observation}\label{obs:red_PREFMATCH}
The pattern $P$ has a circular $k$-edit occurrence in the text $T$ if and only if $\LPref_{k'}[j]+\LSuf_{k-k'}[j] \geq m$
holds for some $j\in [0\dd n]$ and $k'\in [0\dd k]$.
\end{observation}

Unfortunately, the $\LPref_{k'}$ and $\LSuf_{k-k'}$ arrays alone do not give any handle to the starting positions of the underlying circular $k$-edit occurrences. Thus, we also use the following auxiliary lemma whose proof builds upon the Landau--Vishkin algorithm~\cite{DBLP:journals/jal/LandauV89}.

\newcommand{\tbd}{\mathit{LV}}
\begin{restatable}{lemma}{LV}\label{lem:compact-representation-of-ed}
Given a text $T$ of length $n$, a pattern $P$ of length $m$, and an integer $k>0$,
we can compute, in $\Oh(k^2)$ time in the \pillar model, an $\cO(k^2)$-size representation of the edit distances between all pairs of
prefixes of $T$ and $P$ that are at edit distance at most~$k$, that is, the set 
\[\tbd := \{(a,b,d) \in [0 \dd n]\times [0 \dd m] \times [0 \dd k] : \ed(T[0 \dd a), P[0 \dd b))=d \leq k \}.\]
Our representation of $\tbd$ consists of $\cO(k^2)$ sets of the form $\{(a+\Delta,b+\Delta,d) : \Delta \in [0 \dd x) \}$. 
\end{restatable}
\begin{proof}
We rely on the Landau--Vishkin algorithm~\cite{DBLP:journals/jal/LandauV89},
which is based on the fact that the edit distance of two strings does not decrease if we append a letter to each of them.
Thus, $\tbd$ can be decomposed into $\cO(k^2)$ sets of the form $\{(a+\Delta,b+\Delta,d) : \Delta \in [0 \dd x) \}$: $\cO(k)$ of them for each of the central $2k+1$ diagonals of the standard dynamic programming matrix; see \cref{fig:ed_diag_notation}(a).
The Landau--Vishkin algorithm computes exactly this partition of the (prefixes of the) main $2k+1$ diagonals into intervals.
The time complexity of this algorithm is dominated by the time required to answer $\cO(k^2)$ $\LCP$ queries.
\end{proof}

As a warm-up application of \cref{lem:compact-representation-of-ed}, we provide a procedure that, for any $j\in [0\dd n]$,
computes all the values $\LPref_{0}[j],\ldots,\LPref_{k}[j]$.

\begin{corollary}\label{cor:lpref}
  Given a text $T$ of length $n$, a pattern $P$, an integer $k>0$, and a position $j\in [0\dd n]$,
  we can compute, in $\Oh(k^2)$ time in the \pillar model, the values $\LPref_{k'}[j]$ for all $k'\in [0\dd k]$.
\end{corollary}
\begin{proof}
  We build the representation of $\tbd$ from \cref{lem:compact-representation-of-ed} for the strings $T[j\dd n)$ and $P$, and threshold $k$.
  Let us observe that $\LPref_{k'}[j] = \max\{b : (a,b,d)\in \tbd\text{ and }d\le k'\}$. Thus, we initialize $\LPref_{k'}[j]:=0$ for all $k'\in [0\dd k]$
  and iterate over the elements $\{(a+\Delta, b+\Delta, d) : \Delta \in [0\dd x)\}$ in the compact representation of $\tbd$.
  For each such element, we set $\LPref_{d}[j]$ to the maximum of the original value and $b+x-1$.
  Finally, we iterate over the values $\LPref_{k'}[j]$ for all $k'\in (0\dd k]$ (in the increasing order)
  and set $\LPref_{k'}[j]$ to the maximum of the original value and $\LPref_{k'-1}[j]$.
  The total running time on top of the algorithm of \cref{lem:compact-representation-of-ed} is $\Oh(k^2)$ as the representation of $\tbd$ has size $\Oh(k^2)$.
\end{proof}

\begin{figure}[ht]
\centering
\subfloat[Example for \cref{lem:compact-representation-of-ed}: \\$T=\mathtt{abaaabaa}$, $P=\mathtt{aabaabaa}$, $k=2$.]{
   \resizebox{0.32\textwidth}{!}{\begin{tikzpicture}[scale=0.4]

\filldraw[white](1,2) rectangle (4,-11);

\draw[line width=0.23cm,red,opacity=0.2] (0.5,-0.5)--+(1,-1);

\draw[line width=0.23cm,green,opacity=0.2] (1.5,-1.5)--+(1,-1);
\draw[line width=0.23cm,green,opacity=0.2] (1.5,-0.5)--+(2,-2);
\draw[line width=0.23cm,green,opacity=0.2] (0.5,-1.5)--+(4,-4);

\draw[line width=0.23cm,blue,opacity=0.2] (2.5,-2.5)--+(6,-6);
\draw[line width=0.23cm,blue,opacity=0.2] (3.5,-2.5)--+(2,-2);
\draw[line width=0.23cm,blue,opacity=0.2] (4.5,-5.5)--+(1,-1);
\draw[line width=0.23cm,blue,opacity=0.2] (2.5,-0.5)--+(2,-2);
\draw[line width=0.23cm,blue,opacity=0.2] (0.5,-2.5)--+(5,-5);

\node at (4.5,0.75) {$T$};
\foreach[count=\i from 1] \c in {a,b,a,a,a,b,a,a} {
  \node[above] (t\i) at (\i,-0.5) {$\mathtt{\c}$}; 
  \draw[thin] (\i-0.5,0)--+(0,-8.5);
}
\node at (-0.85,-4.5) {$P$};
\foreach[count=\i from 1] \c in {a,a,b,a,a,b,a,a} {
  \node[left] (p\i) at (0.5,-\i) {$\mathtt{\c}$}; 
  \draw[thin] (0,-\i+0.5)--+(8.5,0);
}

\newcommand{\drawDiag}[3]{
  \foreach[count=\i from=1] \v in {#3} {
    \node at (\i+1+#1,-\i-1-#2) {\v};
  }
}

\drawDiag{0}{0}{0,1,2,2,2,2,2,2}
\drawDiag{1}{0}{1,1,2,2}
\drawDiag{2}{0}{2,2}
\drawDiag{0}{1}{1,1,1,1,2}
\drawDiag{0}{2}{2,2,2,2,2}

\end{tikzpicture}}
    }
    \hspace{0.2cm}
\subfloat[Notation used in the  proof of \cref{lem:report-anchored}.]{
    \resizebox{0.6\textwidth}{!}{\begin{tikzpicture}[scale=0.8,transform shape]

\node (i) at (3,0) {};
\node (i1) at (1.9,0) {};
\node (i2) at (3.5,0) {};
\node (i3) at (2.4,0) {};
\node (i4) at (4.0,0) {};
\node (j) at (5,0) {};
\node (ja) at (7.5,0) {};
\node (p1) at (7,0) {};
\node (a) at (4,0) {};
\node (b) at (4.5,0) {};
\node (bb) at (3.5,0) {};

\draw[thick](0,0) rectangle (8,0.5);

\draw[dotted] ($(i3.center)-(0,2)$) rectangle ($(j.center)-(0,1.5)$);
\draw ($(bb.center)-(0,2)$) rectangle ($(j.center)-(0,1.5)$)  node[midway] {$P_2$};
\draw[dotted] ($(j.center)-(0,2)$) rectangle ($(ja.center)-(0,1.5)$);
\draw ($(j.center)-(0,2)$) rectangle ($(p1.center)-(0,1.5)$)  node[midway] {$P_1$};

\node[above left] at (0,0) {$T$};

\node (al) at (4.5,0) {};

\draw (j.center)--+(0,0.5);
\draw (i.center)--+(0,0.5);
\draw[dotted] (a.center)--+(0,1.5);
\draw[dotted] (i1.center)--+(0,1.5);
\draw[dotted] (i2.center)--+(0,1);

\draw[dotted] ($(i3.center)-(0,3)$)--+(0,1);
\draw[dotted] ($(b.center)-(0,3)$)--+(0,1);

\node[above right,yshift=0.5cm] at (i) {$i$};
\node[above right,yshift=0.5cm] at (j) {$j$};

\draw[decorate,decoration={brace,amplitude=4pt}] 
  ($(j.center)-(0,0.1)$)--($(a.center)-(0,0.1)$)
  node[midway, below, yshift=-0.1cm] {$a$};

\draw[decorate,decoration={brace,amplitude=4pt}] 
  ($(a.center)-(0,0.1)$)--($(i.center)-(0,0.1)$)
  node[midway, below, yshift=-0.1cm] {$\Delta$};

\draw[decorate,decoration={brace,amplitude=4pt,mirror}] 
  ($(a.center)+(0,1.5)$)--($(i1.center)+(0,1.5)$)
  node[midway, above, yshift=0.1cm] {$I$};

\draw[decorate,decoration={brace,amplitude=4pt,mirror}] 
  ($(i2.center)+(0,1)$)--($(i1.center)+(0,1)$)
  node[midway, above, yshift=0.1cm] {$I'$};

\draw[decorate,decoration={brace,amplitude=4pt}] 
  ($(j.center)-(0,2.1)$)--($(b.center)-(0,2.1)$)
  node[midway, below, yshift=-0.1cm] {$b$};

\draw[decorate,decoration={brace,amplitude=4pt}] 
  ($(b.center)-(0,2.1)$)--($(bb.center)-(0,2.1)$)
  node[midway, below, yshift=-0.1cm] {$\Delta$};

\draw[decorate,decoration={brace,amplitude=4pt}] 
  ($(ja.center)-(0,2.1)$)--($(j.center)-(0,2.1)$)
  node[midway, below, yshift=-0.1cm] {$\alpha$};

\draw[decorate,decoration={brace,amplitude=4pt}] 
  ($(b.center)-(0,3)$)--($(i3.center)-(0,3)$)
  node[midway, below, yshift=-0.1cm] {$I$};

\draw[|<->|] ($(ja.center)-(0,3.5)$)--($(i4.center)-(0,3.5)$)
  node[midway, below, yshift=-0.1cm] {$m$};

\draw[decorate,decoration={brace,amplitude=4pt}] 
  ($(i4.center)-(0,3.5)$)--($(i3.center)-(0,3.5)$)
  node[midway, below, yshift=-0.1cm] {$I'$};

\draw[dotted] ($(j.center)$)--($(j.center)-(0,2.1)$);
\end{tikzpicture}}}
\caption{Illustrations for~\cref{lem:compact-representation-of-ed,lem:report-anchored}.}\label{fig:ed_diag_notation}
\end{figure}

Our main application of \cref{lem:compact-representation-of-ed}, however, is the following result, which allows reporting the starting positions $i$
of all cyclic $k$-edit occurrences satisfying \cref{obs:ed_cl} for a given $j$.

\begin{lemma}\label{lem:report-anchored}
Given a text $T$ of length $n$, a pattern $P$ of length $m$, an integer $k>0$,
and a position $j\in [0\dd n]$, one can compute, in $\Oh(k^2)$ time in the \pillar model, all positions $i\in [0\dd j]$ such that
$\ed(T[i \dd j), P_2) + \ed(T[j \dd r),P_1) \leq k$ holds for some $r\in [j\dd n]$ and some decomposition $P=P_1P_2$.
The output is represented as a union of $\cO(k^2)$ (possibly intersecting) intervals.
\end{lemma}

\begin{proof}
We start with the calculation of the compact representation of $\tbd$ from \cref{lem:compact-representation-of-ed} for the reversals of strings $T[0\dd j)$ and $P$, and threshold $k$.
Next, for each element $\{(a+\Delta, b+\Delta, d) : \Delta \in I\}$ of this compact representation, we will
calculate an interval $I'\subseteq I$ such that positions from $\{ j-a-\Delta: \Delta \in I'\}$ are in $\cyc{}{k}(P, T)$.
The notation used in this proof is illustrated in \cref{fig:ed_diag_notation}(b).

From the definition of $\tbd$, we know that, for every $\Delta \in I$, we have $\ed(T[i\dd j), P_2)=d$, where $i=j-a-\Delta$ and $P_2=P[m-b-\Delta\dd m)$.
All we have to do is to verify if there exists a position $r$ in $T$ such that $\ed(T[j\dd r), P_1)\le k-d$, where $P_1 = P[0\dd m-b-\Delta)$.

For this, we use \cref{cor:lpref} to compute the values $\LPref_{k'}[j]$ for all $k'\in [0\dd k]$. 
Then, the maximal possible length of $P_1$ (within our edit distance budget)
is $\alpha:=\LPref_{k-d}[j]$.
Now, we need to define $I'$ in such a way that it corresponds to pairs $(P_1,P_2)$ with total length $m$, that is, we set
$I':=\{ \Delta \in I : b + \Delta + \alpha \ge m \}$.

The most time-consuming parts of this procedure are the applications of \cref{lem:compact-representation-of-ed,cor:lpref}, taking $\Oh(k^2)$ time in the \pillar model.
\end{proof}

The algorithm behind \cref{thm:ed_rep}, whose statement is repeated below, is a simple application of \cref{lem:report-anchored}. 

\reportingEditCPM*
\begin{proof}
By \cref{obs:ed_cl}, it suffices to apply \cref{lem:report-anchored} for all $j\in [0\dd n]$
and report all positions in the union of the obtained $\Oh(nk^2)$ intervals.
The primitive \pillar operations can be implemented in $\Oh(1)$ time after $\Oh(n)$-time preprocessing (\cref{the:standardPILLAR}),
so the applications of \cref{lem:report-anchored} take $\Oh(n+n\cdot k^2)=\Oh(nk^2)$ time in total.
The union of the resulting intervals can be computed in $\Oh(nk^2)$ time using bucket sort.
\end{proof}

As for \cref{thm:ed_dec}, we also need efficient construction of the $\LPref_{k'}[0\dd n]$ and $\LSuf_{k'}[0\dd n]$ arrays for all $k'\in [0\dd k]$.
Both~\cite{DBLP:journals/siamcomp/LandauMS98} and \cref{cor:lpref} yield $\Oh(nk^2)$-time algorithms.
However, \cref{thm:allkPREFMATCH}, proved in the next section, improves this running time to $\Oh(nk\log^2 k)$.
Actually, we use \cref{lem:allkPREFMATCH} that implies \cref{thm:allkPREFMATCH} (see \cref{sec:AllkPREFMATCH}) to first obtain an algorithm for the decision version of $k$-Edit CPM in the \pillar model; such an algorithm has already proved to be useful (see \cite{stacs24}).

\begin{lemma}\label{lem:forSTACS}
Given a text $T$ of length $n$, a pattern $P$ of length $m$, an integer $k>0$,
and an interval $J \subseteq [0\dd n]$ of length $\Oh(k)$, one can compute, in $\Oh(k^2\log^2 k)$ time plus $\Oh(k^2)$ \pillar operations, any position $i$ such that
$\ed(T[i \dd j), P_2) + \ed(T[j \dd r),P_1) \leq k$ holds for some $j \in J$ and $r\in [j\dd n]$ and some decomposition $P=P_1P_2$, or state that no such position $i$ exists.
\end{lemma}
\begin{proof}
By \cref{obs:red_PREFMATCH}, a requested position $i$ exists if and only if $\LPref_{k'}[j]+\LSuf_{k-k'}[j] \ge m$
holds for some $j\in J$ and $k'\in [0\dd k]$. By \cref{lem:allkPREFMATCH}, $\LPref_{k'}[j]$ for all $j \in J$ and $k'\in [0\dd k]$ can be computed in $\Oh(k^2\log^2 k)$ time plus $\Oh(k^2)$ $\LCP$ queries on suffixes of $P$ and $T$, which are \pillar operations.
As for computing $\LSuf_{k-k'}[j]$ for $j \in J$, it suffices to use \cref{lem:allkPREFMATCH} for the reversed strings and then reverse the resulting arrays.
Overall, it takes $\Oh(k^2\log^2 k)$ time plus $\Oh(k^2)$ \pillar operations to decide if a requested position $i$ exists.

If the answer is positive, we also need to report a witness occurrence $i$.
For this, we fix an arbitrary $j\in J$ such that $\LPref_{k'}[j]+\LSuf_{k-k'}[j] \ge m$ holds for some $k'\in [0\dd k]$
and run the algorithm of \cref{lem:report-anchored}.
By \cref{obs:ed_cl}, this call is guaranteed to report at least one position.
The cost of the call in the \pillar model is $\Oh(k^2)$.
\end{proof}

\decisionEditCPM*
\begin{proof}
We invoke the algorithm underlying \cref{lem:forSTACS} for a collection of $\Oh(n/k)$ intervals of length~$k$ that cover $[0 \dd n]$. We recall that \pillar operations can be performed in constant time after an $\cO(n)$-time preprocessing (cf.\ \cref{the:standardPILLAR}).
The overall running time is thus $\Oh(n+ (n/k) \cdot k^2\log^2 k)=\Oh(nk\log^2 k)$.
\end{proof}

\section{An Algorithm for \texorpdfstring{\allkep:}{All-k-Edit-PrefMatch:} Proof of \texorpdfstring{Theorem~\ref{thm:allkPREFMATCH}}{Theorem~4}}\label{sec:AllkPREFMATCH}
In this section, we show how to compute, given a position $q$ in the text, the values $\LPref_{k'}[i]$ for all $k' \in [0 \dd k]$ and $i \in [q \dd q+k)\cap [0\dd n]$ in $\Oh(k^2\log^2 k)$ time in the \pillar model.
This will yield the desired solution to the \allkep problem by taking values of $q$ which are multiples of $k$.
The next two subsections deliver the necessary technical tools related to Monge matrices.

\subsection{Dynamic and Persistent Sub-row Minimum Queries on Monge Matrices}\label{subsec:Monge}
In this section, we extend a known result on answering submatrix minimum queries on Monge matrices~\cite{DBLP:journals/talg/KaplanMNS17} (see also~\cite{DBLP:journals/talg/GawrychowskiMW20}) to a dynamic and persistent setting: the matrix changes by sub-column increments and we need to be able to query all previously created matrices. Sub-row queries, which were considered as a simpler case in~\cite{DBLP:journals/talg/KaplanMNS17,DBLP:journals/talg/GawrychowskiMW20}, are sufficient for our purposes.

We recall that a matrix $M$ is a \emph{Monge matrix} if, for every pair of rows $i<i'$ and every pair of columns $j < j'$, \[M[i,j]+M[i',j'] \le M[i,j']+M[i',j].\]
A matrix $M$ is called an \emph{inverse Monge matrix} if, for every pair of rows $i<i'$ and every pair of columns $j < j'$, we have $M[i,j]+M[i',j'] \ge M[i,j']+M[i',j]$.

\begin{example}\label{ex:Monge}
The following matrix is Monge. Consider the ``boxed'' elements; we have red $\le$ blue, i.e., $A[2,1]+A[4,3]\,\le\, 
A[4,1]+A[2,3]$.
\[
A=\begin{bmatrix}
10 & 17 & 13 & 28 & 23 \\
17 & 22 & 16 & 29 & 23 \\
24 & \fbox{\textcolor{red}{28}} & 22 & \fbox{\textcolor{blue}{34}} & 24 \\
11 & 13 & 6  & 17 & 7  \\
45 & \fbox{\textcolor{blue}{44}} & 32 & \fbox{\textcolor{red}{37}} & 23 \\
36 & 33 & 19 & 21 & 6  \\
75 & 66 & 51 & 53 & 34
\end{bmatrix}
\]
By negating all  entries of $A$, we obtain an inverse Monge matrix $M = -A$.
\end{example}

Let $M$ be an $r \times c$ matrix. 
For each column $j$, we define a \emph{column function} $f_j^M: [0 \dd r) \rightarrow \mathbb{R}$ such that $f_j^M(x)=M[x,j]$ for $x \in [0\dd r)$.
The inverse Monge inequalities imply the following property.

\begin{lemma}[{\cite[Section 2]{DBLP:journals/talg/KaplanMNS17}}]\label{lem:diff}
Let $M$ be an inverse Monge matrix. For every $j,j' \in [0 \dd c)$, if $j<j'$, then the function $f_j^M-f_{j'}^M$ is weakly decreasing.
\end{lemma}

An \emph{upper envelope} of functions $f_1,\dots,f_p : [0 \dd r) \to \mathbb{R}$ is a function 
\[g(x)=\max \{ f_k(x)\,:\, k\in [1\dd p]\}.\]
By $\E_{j,j'}^M$ we denote the upper envelope of column functions $f_j^M,f_{j+1}^M,\dots,f_{j'}^M$. In what follows we may omit the superscript $M$ if it is clear from the context.

The next observation follows from the definition of inverse Monge matrix (and \cref{lem:diff}).
\begin{observation}\label{obvious}
  Let $M$ be an $r \times c$ inverse Monge matrix, 
 $0 \le a \le s < b < c$, and $M'$ be an\\ $r\times 2$ matrix, where
 $M'[i,0]=\E_{a,s}(i),\ M'[i,1]=\E_{s+1,b}(i)$ for $i\in [0\dd r)$.
 Then 
 \begin{enumerate}[(a)]
 \item $M'$ is inverse Monge.
\item There exists $x_0 \in [0 \dd r]$ such that 
\[\forall\, x \in [0 \dd x_0)\ \ M'[x,0]\ge M'[x,1],\quad \forall\, x \in [x_0 \dd r)\ \ M'[x,0]<M'[x,1].
\]
\item
There exists $x_0 \in [0 \dd r]$ such that
\[\forall\, x \in [0 \dd x_0)\ \ \E_{a,b}(x)=\E_{a,s}(x) \ge \E_{s+1,b}(x),\quad
\forall\, x \in [x_0 \dd r)\ \ \E_{a,b}(x)=\E_{s+1,b}(x)>\E_{a,s}(x).\]
\end{enumerate}
\end{observation}

\paragraph{Persistent data structures.}
We use the standard technique for making rooted trees (e.g., BSTs) persistent (cf.~\cite{DBLP:journals/jcss/DriscollSST89}), that is, whenever the tree is updated due to an operation top-down (say, an insertion in a BST), copies of all nodes that are changed throughout the operation are created.
Some pointers from the copied nodes may lead to nodes from a previous instance. The \emph{state} of the data structure after the $t$-th operation is represented by (a copy of) the root created by the $t$-th operation. If the tree has height $h$, updating the nodes on a root-to-node path results in the creation of $\Oh(h)$ nodes. Said duplication of nodes does not affect the asymptotic running times of our algorithms.

The next lemma is the main result of this section. Note that while we support update operations that add an arbitrary value to all entries in a sub-column (under a guarantee that the resulting matrix still satisfies the Monge property), incrementing by one would suffice for our needs. Moreover, partial persistency (that is, allowing the next matrix to be created only based on the most recent one) would also be sufficient for us.
We present a solution for a general variant of the problem, as it is equally fast and no more complicated.

\begin{lemma}\label{lem:Monge}
Let $M_0$ be an $r \times c$ Monge matrix. Consider a sequence of operations of the following two types:
\begin{itemize}
    \item \emph{Sub-column addition}: Create a new Monge matrix $M_{t}$ obtained from a given previous matrix $M_{t'}$ by adding a given value to the entries in a given sub-column.
    \item \emph{Sub-row minimum query}: Extract the minimum entry in a given sub-row of a specified previously created matrix $M_t$.
\end{itemize}
After an $\Oh(cr)$-time preprocessing on $M_0$,
we can support both types of operations in $\Oh(\log c \log r)$ time. Moreover, any sub-row minimum query that asks for the minimum of a full row (that is, a row minimum query) can be answered in $\Oh(\log r)$ time.
\end{lemma}
\begin{proof}
Henceforth, we assume that $r+1$ is a power of two. This can be achieved by noting that we can duplicate the last row as many times as necessary without affecting the Monge property.
By negating all the elements of a Monge matrix, we obtain an inverse Monge matrix.
Thus, our goal is to perform sub-column additions (of negated values) and answer sub-row \emph{maximum} queries on a fully persistent inverse Monge matrix. Let $M_0,\dots,M_p$ denote the sequence of inverse Monge matrices created within the process.

\medskip\noindent {\bf Data structure.}
A function $f:[0 \dd r) \rightarrow \mathbb{R}$ will be stored as a perfectly balanced binary tree~$\T(f)$. Each node of the tree corresponds to an interval of arguments. The root corresponds to the interval $[0 \dd r)$. Each node corresponds to an interval $[p \dd q] \subseteq [0 \dd r)$, where $q-p+2$ is an integer power of~$2$. Unless $p=q$, the corresponding node has children corresponding to intervals $[p \dd (p+q)/2-1]$ and $[(p+q)/2+1 \dd q]$, respectively. Each node stores a value $x$ and a weight $\delta$. The \emph{actual value} of a node $v$ is the sum of the value of $v$ and the weights of nodes on the $v$-to-root path of $\T(f)$. The actual value of the node corresponding to interval $[p \dd q]$ is to be equal to $f((p+q)/2)$; we also say that this node represents $f((p+q)/2)$.

We use such a tree representation $\T(f)$ for each column function $f$, storing it as a persistent data structure. The weights will be updated in a lazy manner. That is, when accessing value~$f(x)$ represented by a node $v$ of $\T(f)$, we make sure that the weights on the $v$-to-root path are all equal to 0 by propagating the non-zero values down the tree; then, $f(x)$ is simply the value of $v$. (This procedure is used, in particular, in sub-row maxima queries.)

\newcommand{\TT}{\mathbf{T}}
For a matrix $M$, let $\TT(M)$ be a static balanced binary tree formed over the set of columns of~$M$; the subtree $\TT_v(M)$ rooted at a node $v$ of $\TT(M)$ corresponds to a range of consecutive columns of~$M$.
Each node $v$ of $\TT(M)$ stores the upper envelope $\E_{a,b}^M$ of the column functions corresponding to the column range $[a \dd b]$ in $\TT_v(M)$---let us denote it as $\E_v^M$. 

Each upper envelope $\E_{a,b}$ in $\TT(M)$ will be stored in a persistent manner by its tree representation $\T(\E_{a,b})$. If $a=b$, this is simply $\T(f_a)$. Otherwise, $\T(\E_{a,b})$ will be constructed recursively based on the representations $\T(\E_{a,s})$ and $\T(\E_{s+1,b})$ of upper envelopes in the children of the node of~$\TT(M)$ that contains $\T(\E_{a,b})$, as we discuss below.

Let $f=\E_{a,b}$, $f_1=\E_{a,s}$, and $f_2=\E_{s+1,b}$. By \cref{obvious},
there exists $x_0 \in [0 \dd r]$ such that 
\[f(x)=f_1(x) \ge f_2(x)\ \mbox{for}\ x \in [0 \dd x_0)\ \mbox{and}\ 
f(x)=f_2(x)>f_1(x)\ \mbox{for}\ x \in [x_0 \dd r).\]
Our computation of $\T(f)$ resembles binary search for $x_0$. Let $x$ be such that the roots of $\T(f_1)$ and~$\T(f_2)$ represent $f_1(x)$ and $f_2(x)$, respectively. 

If $f_1(x) \ge f_2(x)$, then $x < x_0$. In this case, the root of $\T(f_1)$ represents $f(x)$ and the left subtree of $\T(f)$ is the left subtree of $\T(f_1)$, while the right subtree of $\T(f)$ is computed recursively from the right subtrees in $\T(f_1)$ and $\T(f_2)$. 

On the other hand, if $f_1(x) < f_2(x)$, then $x \ge x_0$. Then, the root of $\T(f_2)$ represents $f(x)$ and the right subtree of $\T(f)$ is the right subtree of $\T(f_2)$, while the left subtree of $\T(f)$ is computed recursively from the left subtrees in $\T(f_1)$ and~$\T(f_2)$. 

Before each recursive call, if the weight in the root of the considered subtree of $\T(f_i)$, for $i \in \{1,2\}$, was non-zero, we add it to the children of the root and set it to zero. (For persistency, we create new nodes with the updated weights.) Overall, $\T(f)$ is computed in $\Oh(\log r)$ time and stored in just $\Oh(\log r)$ extra space by reusing subtrees of $\T(f_1)$ and $\T(f_2)$. Observe that we ensure, for $i \in \{1,2\}$, that the weights of all the ancestors of the root of any reused subtree are equal to zero, and hence the actual values of the nodes of $\T(f)$ are correct.

\medskip\noindent {\bf Initialisation.}
For each column function $f_j^{M_0}$, we create its tree representation $\T(f_j^{M_0})$ by definition, with all weights equal to 0, in $\Oh(r)$ time, for a total of $\Oh(cr)$ time. We compute all upper envelopes in $\TT(M_0)$ bottom-up. The tree representation of each upper envelope of a node~$v$ in $\TT(M_0)$ takes $\Oh(\log r)$ time to compute from the upper envelopes in the children of $v$. Tree $\TT(M_0)$ contains $\Oh(c)$ nodes, so the upper envelopes for non-leaf nodes require $\Oh(c \log r)$ time to be computed. In total, the initialisation takes $\Oh(cr)$ time.

\medskip\noindent {\bf Sub-column addition.}
In each such operation, a matrix $M_{t'}$, a column $j$, a range of rows $[i\dd i']$, and a value $\alpha$ are specified. We are to create a next matrix $M_t$ such that $M_t[x,y]=M_{t'}[x,y]+\alpha$ if $(x,y) \in [i \dd i'] \times \{j\}$ and $M_t[x,y]=M_{t'}[x,y]$ otherwise.

First, we access the tree representation of the column function $f_j^{M_{t'}}$ via the leaf of $\TT(M_{t'})$ that corresponds to the column $j$. We compute, in $\Oh(\log r)$ time, a set of $\Oh(\log r)$ so-called \emph{canonical nodes} in $\T(f_j^{M_{t'}})$ whose intervals are disjoint and cover $[i \dd i']$. During the traversal, the weights are propagated so that there are no non-zero weights (weakly) above the canonical nodes. We then set the weight of each canonical node to $\alpha$. Each operation on a node creates a new node for persistence. This step takes $\Oh(\log r)$ time.

Next, all nodes in $\TT(M_{t'})$ whose column range contains $j$ are updated, bottom-up from the leaf that stores $f_j^{M_{t'}}$. In each of them, the upper envelope is computed from scratch as described above in $\Oh(\log r)$ time, for a total of $\Oh(\log c \log r)$ time. Again, for every visited node, a new node is created for persistence. The root of the resulting tree becomes the handle for $\TT(M_t)$, the tree representation of the newly created matrix $M_t$.

\medskip\noindent {\bf Sub-row maximum query.}
A query takes as input a handle to matrix $M_t$, a row $i$, and a range of columns $[j\dd j']$. To answer it, we compute, in $\Oh(\log c)$ time, $\Oh(\log c)$ canonical nodes of $\TT(M_t)$ whose column ranges are disjoint and cover $[j\dd j']$.
For each such canonical node $u$, we compute $\E_u^{M_t}(i)$ in $\Oh(\log r)$ time using the tree representation $\T(\E_u^{M_t})$, as described above. The output is the maximum of these values. The total query time is $\Oh(\log c \log r)$.

If the query concerns a full row, it suffices to query the root of $\TT(M_t)$. Thus, the query time drops to $\Oh(\log r)$.
\end{proof}

\subsection{Online Seaweed Combing}\label{subsec:seaweeds}
The \emph{deletion distance} $\ddist(U,V)$ of two strings $U$ and $V$ is the minimum number of letter insertions and deletions required
to transform $U$ to $V$; in comparison to edit distance, substitutions are not directly allowed (they can be simulated at cost 2 by an insertion and a deletion).

\begin{definition}\label{def:align_graph}
For strings $U$ and $V$ and an interval $I$, we define the \emph{alignment graph} $G(U,V,I)$ as the weighted undirected graph with vertex set $\mathbb{Z}^2$ and the following edges for each $(x,y) \in \mathbb{Z}^2$:
\begin{itemize}
    \item $(x,y) \stackrel{1}{\longleftrightarrow} (x+1,y)$,
    \item $(x,y) \stackrel{1}{\longleftrightarrow} (x,y+1)$,
    \item 
    $(x,y) \stackrel{0}{\longleftrightarrow} (x+1,y+1)$, present unless $x \in [0\dd |U|)$, $y \in [0 \dd |V|)$, $U[x] \ne V[y]$, and $y-x \in I$.
\end{itemize}
\end{definition}

For an alignment graph $G(U,V,I)$, there are $|I|$ diagonals where diagonal edges may be missing.
Intuitively, everything outside these diagonals is considered a match.
See \cref{fig:align} for an illustration of an alignment graph.

\begin{figure}[ht]
\captionsetup{singlelinecheck=off}
\centering
\begin{tikzpicture}[scale=0.8]

\foreach \x/\z in {-1/0, 0/1, 1/2, 2/3, 3/4, 4/5, 5/6}{
   \foreach \y in {0,1,2,3,4,5}{
	\draw[thick] (\x,\y) -- (\x+1,\y-1);
   }
}

\filldraw[black!10!white] (0,2) -- (0,4) -- (3.1,4) -- (5,2.1) -- (5,0) -- (1.9,0) -- (0,1.9);

\draw[green!30!white,line width=1.5mm] (0,4) -- (1,3) -- (2,2) -- (3,2) -- (4,1) -- (5,1) -- (5,0);
  
\foreach \x in {-1,0, 1,2,3,4,5,6}{
   \foreach \y in {-1,0, 1,2,3,4,5}{
   \filldraw(\x,\y) circle (1.5pt);
   }
}

\foreach \x/\z in {-1/0, 0/1, 1/2, 2/3, 3/4, 4/5, 5/6}{
   \foreach \y in {-1,0,1,2,3,4,5}{
	\draw[thick] (\x,\y) -- (\z,\y);
   }
}

\foreach \x/\z in {-1/0, 0/1, 1/2, 2/3, 3/4, 4/5}{
   \foreach \y in {-1,0,1,2,3,4,5,6}{
	\draw[thick] (\y,\z) -- (\y,\x);
   }
}

\foreach \x/\y in {0/4, 3/4, 0/2, 3/2, 4/3, 3/4, 1/3, 2/1, 0/1, 4/4}{
	\draw[thick] (\x,\y) -- (\x+1,\y-1);
   }
   
\draw (-1.2,5) node[above] {$(-1,-1)$};
\draw (6.2,5) node[above] {$(-1,6)$};

\draw (6.2,-1) node[below] {$(5,6)$};
\draw (-1.2,-1) node[below] {$(5,-1)$};

\foreach \x/\c in {0.5/a, 1.5/b, 2.5/c, 3.5/a, 4.5/b}{
   \draw (\x,4) node[above=-0.1cm] {$\texttt{\c}$};
}

\foreach \x/\c in {0.5/c,1.5/a, 2.5/b, 3.5/a}{
   \draw (0,\x) node[left=-0.1cm] {$\texttt{\c}$};
}
\end{tikzpicture}
\caption{
Consider $U=\texttt{abac}$ and $V=\texttt{abcab}$ and $I = [-2 \dd 3]$.
The figure depicts the subgraph of the alignment graph $G(U,V,I)$ induced by the set of vertices $[-1 \dd 5]\times [-1 \dd 6]$ (edge weights omitted).
Note that diagonal edges are only missing in the shaded portion of the shown graph.
A shortest $(0,0)$-to-$(4,5)$ path is highlighted in green; its length equals $\delta_D(U,V)=3$. It follows that $U$ can be transformed to $V$ as follows: $\texttt{aba\textcolor{red}{c}} \rightarrow \texttt{ab\textcolor{green!60!black}{c}a\textcolor{green!60!black}{b}}$.)}\label{fig:align}
\end{figure}

\begin{observation}[see {\cite[Lemma 8.5]{https://doi.org/10.48550/arxiv.2204.03087}}]\label{obs:alignment}
For all \fragments $U[x\dd x')$ and $V[y\dd y')$ of $U$ and $V$, respectively, $\ddist(U[x\dd x'),V[y \dd y'))$ 
is the length of the shortest $(x,y)\leadsto (x',y')$ path in $G(U,V,\mathbb{Z})$. 
\end{observation}

Let us fix a text $T$ of length $n$, a pattern $P$ of length $m$, and a positive integer $k$ for this subsection.

\begin{definition}\label{def:Dqz}
For integers $q \in [0\dd n)$ and $z \in [0 \dd m]$, let $G_{q,z} := G(P[0 \dd z),T,[q-k\dd q+2k))$.
For each $z \in [0 \dd m]$, let
 $D_{q,z}[0 \dd 3k+1,0 \dd 3k+1]$ be a $(3k+2) \times (3k+2)$ \emph{distance matrix} such that $D_{q,z}[a,b]$ is the length of the shortest path between $(0,a+q-k-1)$ and $(z,z+b+q-k-1)$
in the alignment graph $G_{q,z}$.
\end{definition}

The distance matrix directly stores deletion distances between the required substrings of $P$ and~$T$ provided that these distances do not exceed $k$.
This is formally stated in the lemma below, which is a variation of \cref{obs:alignment}. 
The lemma follows from the well-known diagonal-band argument that originates from a work of Ukkonen~\cite{DBLP:journals/iandc/Ukkonen85} and has been used, for instance, by Landau and Vishkin~\cite{DBLP:journals/jal/LandauV89} towards their $\cO(nk)$-time algorithm for pattern matching under edit distance.

\begin{lemma}\label{lem:interface:Ka}
Let 
$z \in [0 \dd m]$, $i \in [q\dd q+k)$, $j \in [q-k-1 \dd q+2k] \cap [i-z\dd n-z]$, and $k'\in [0\dd k]$. Then,
\[\ddist(P[0\dd z),T[i\dd j+z))\le k'\quad \Longleftrightarrow\quad D_{q,z}[i-q+k+1,j-q+k+1] \le k'.\]
\end{lemma}
\begin{proof}
Let us note that $i-q+k+1,j-q+k+1 \in [0 \dd 3k+1]$ and hence $D_{q,z}[i-q+k+1,j-q+k+1]$ is well-defined.

    First, consider the case when $j \not\in [i-k \dd i+k]$. In this case,  $D_{q,z}[i-q+k+1,j-q+k+1] > k$. This is because any $(0,i)$-to-$(z, z+j)$ path in $G_{q,z}$ contains at least $|j-i| > k$ non-diagonal edges, each contributing a unit to the length of the path.
    When $j \not\in [i-k \dd i+k]$, we also have
    \[\ddist(P[0\dd z),T[i\dd j+z))\geq |\,|P[0\dd z)| - |T[i\dd j+z)|\,| = |j-i| > k.\]
    We can thus henceforth focus on the case when $j \in [i-k \dd i+k]$.

\begin{enumerate}
    \item\label{e1} Observe that $G(P, T, \mathbb{Z})$ is a subgraph of $G(P,T,I)$ (for any interval $I$). Consequently, \cref{obs:alignment} and the definition of $D_{q,z}$
    yield \[D_{q,z}[i-q+k+1,j-q+k+1]\le \ddist(P[0\dd z),T[i\dd j+z)).\]
    We note that condition $j \in [i-z \dd n-z]$ ensures that $T[i \dd j+z)$ is a valid \fragment of $T$ (possibly empty), so \cref{obs:alignment} can indeed be applied.

    \smallskip\noindent 
    \item\label{e2}
    Moreover, if $I\supseteq [y-x-k \dd y-x+k]$,
    then any $(x,y)\leadsto (x',y')$ path in $G(P,T,I)$ of weight at most $k$ contains only vertices $(x'',y'')$ with $y''-x'' \in I$, and thus the same path is also present in $G(P, T, \mathbb{Z})$.
    Since $[i-k \dd i+k] \subseteq [q-k\dd q+2k)=I$, then $D_{q,z}[i-q+k+1,j-q+k+1] \le k'\le k$ implies $\ddist(P[0\dd z),T[i\dd j+z)) \le k'$ due to \cref{obs:alignment}.\qedhere
\end{enumerate}
\end{proof}

Pairwise distances of vertices that lie on the outer face of a connected undirected planar graph satisfy the \emph{Monge property}.

\begin{fact}[{\cite[Section 2.3]{FR06}}]\label{fct:monge}
    Consider a connected planar graph $G$ with non-negative edge
    weights. For vertices
    $u_0,\ldots,u_{p-1},v_{q-1},\ldots,v_0$ that (in this cyclic order) lie on the outer
    face
    of $G$, the $p\times q$ matrix $D$ with $D[i,j]=\dist_G(u_i,v_j)$ is a Monge matrix.
\end{fact}

Now the planarity of $G_{q,z}$ and the fact that all vertices of the considered shortest paths lie in $[0\dd z] \times [-k-1\dd n+m+2k]$ imply the following.

\begin{observation}\label{obs:dt_monge}
For every $q \in [0 \dd n)$ and $z \in [0 \dd m]$, the matrix $D_{q,z}$ is a Monge matrix.
\end{observation}

The remainder of this section is devoted to computing an efficient representation of the sequence of matrices $D_{q,z}$, for $z \in [0 \dd m]$, via so-called permutation matrices. A permutation matrix is a square matrix over $\{0,1\}$ that contains exactly one 1 in each row and in each column.
A permutation matrix~$\Pi$ of size $s \times s$ can be represented, for any $d \in \mathbb{Z}$, by the permutation~$\pi:[d \dd d+s)\to [d\dd d+s)$ such that $\Pi[i-d,j-d]=1$ if and only if $\pi(i)=j$.

For two permutations $\pi_1$, $\pi_2$ and their corresponding permutation matrices $\Pi_1$, $\Pi_2$, by $\Delta(\pi_1,\pi_2)=\Delta(\Pi_1,\Pi_2)$ we denote a sequence of swaps of neighbouring columns that transforms $\Pi_1$ into~$\Pi_2$.

A swap of columns $j$ and $j+1$ is called \emph{ordered} if $\Pi_1[i,j]=1=\Pi_1[i',j+1]$ holds for some rows~$i<i'$.

For an $s\times s$ matrix $A$, we denote by $A^\Sigma$ an $(s+1)\times (s+1)$ matrix such that \[A^{\Sigma}[i,j]=\sum_{i' \ge i}\sum_{j'<j} A[i',j']\ \text{ for }i,j\in [0\dd s].\]

\begin{figure}[htpb!]
    \centering
    \begin{tikzpicture}
        \draw[thick] (0,0) rectangle (3,3);
        \draw[fill = brown!30!white] (0,0) rectangle (2,2.2);
        \draw[dashed] (2,3) -- (2,2.2);
        \draw (-.25,1.5) node[left] {$A \, =$};
        \draw (0,2.1) node[left] {$i$};
        \draw (2.1,3) node[above] {$j$};
    \end{tikzpicture}
    \caption{A matrix $A$ is shown. The value $A^{\Sigma}[i,j]$ is equal to the sum of the entries of $A$ in the shaded area.}
\end{figure}

We are now ready to state the main result of this subsection.

\begin{restatable}{lemma}{mainseaweed}\label{lem:interface:Kc}
For each $q \in [0 \dd n)$ and $z \in [0 \dd m]$, there is a $(3k+1)\times (3k+1)$ permutation matrix~$\Pi_{q,z}$ such that \[\forall_{i,j\in [0\dd 3k+1]}\; D_{q,z}[i,j] = 2\Pi_{q,z}^{\Sigma}[i,j]+i-j.\]
Moreover, $\Pi_{q,0}$ is an identity permutation matrix and there exists a sequence of updates $\Delta(\Pi_{q,0},\Pi_{q,1})$, \dots, $\Delta(\Pi_{q,m-1},\Pi_{q,m})$ that consists of at most $3k(3k+1)/2$ ordered swaps of neighboring columns in total. Its non-empty elements can be computed in $\Oh(k^2 \log \log k)$ time plus $\Oh(k^2)$ $\LCP$ queries on pairs of suffixes of $P$ and~$T$.
\end{restatable}

Before proving the lemma, we introduce useful notations and results from~\cite[Sections 8 and~9]{https://doi.org/10.48550/arxiv.2204.03087} which builds upon the ideas of Tiskin~\cite{https://doi.org/10.48550/arxiv.0707.3619,DBLP:journals/mics/Tiskin08} and in particular his so-called \emph{seaweed combing} technique.

We recall from \cite{https://doi.org/10.48550/arxiv.2204.03087} the definition of an alignment graph for an arbitrary set $M \subseteq \mathbb{Z}^2$ of excluded diagonal edges.

\newcommand{\AG}{\mathsf{AG}}
\begin{definition}[{\cite[Definition 8.4]{https://doi.org/10.48550/arxiv.2204.03087}}]
Given a set $M\subseteq \mathbb{Z}^2$, we define the \emph{general alignment graph} $\AG(M)$ with vertices $\mathbb{Z}^2$ and weighted edges:
\begin{itemize}
    \item $(x,y)\stackrel{1}{\longleftrightarrow} (x+1,y)$ for every $(x,y)\in \mathbb{Z}^2$;
    \item $(x,y)\stackrel{1}{\longleftrightarrow} (x,y+1)$ for every $(x,y)\in \mathbb{Z}^2$;
    \item $(x,y)\stackrel{0}{\longleftrightarrow} (x+1,y+1)$ for every $(x,y)\in \mathbb{Z}^2 \setminus M$.
\end{itemize}
We denote the underlying distance function on $\mathbb{Z}^2$ by $\dist_M$.
\end{definition}

\begin{remark}
In \cite{https://doi.org/10.48550/arxiv.2204.03087}, $(i,j)$ denotes the point with $x$-coordinate $i$ and $y$-coordinate $j$.
Here, consistently with subsequent work, point $(i,j)$ denotes the point at the $i$-th row and the $j$-th column of the grid.
In what follows, we have swapped the coordinates of points in all subsequent definitions and results that we use from \cite{https://doi.org/10.48550/arxiv.2204.03087}.
\end{remark}

For strings $U$ and $V$, an interval $I$, and a threshold $z \in \mathbb{Z}$, we denote
\[\mu(U,V,I,z) = \{(x,y)\,:\,x < z,\,(x,y) \stackrel{0}{\longleftrightarrow} (x+1,y+1) \text{ does not exist in } G(U,V,I)\}.\]

\begin{observation}\label{obs:useme}
For every $z\in [0\dd m]$, we have $\AG(\mu(P,T,I,z)) = G(P[0 \dd z), T, I)$.
\end{observation}
\begin{proof}
It suffices to check which 0-weight diagonal edges are absent in both graphs. The edge $(x,y) \leftrightarrow (x+1,y+1)$ is missing in $\AG(\mu(P,T,I,z))$ if $x<z$ and the edge is missing in $G(P,T,I)$. 

The latter can be stated equivalently as $x \in [0 \dd m)$, $y \in [0 \dd n)$, $y-x \in I$, and $P[x] \ne T[y]$. As $x<z$, the last condition is equivalent to $P'[x] \ne T[y]$ for $P'= P[0 \dd z)$. Hence, the edge is missing in $G(P',T,I)$. 
A proof that edges missing in $G(P', T, I)$ are also missing in $\AG(\mu(P,T,I,z))$ is symmetric.
\end{proof}

\begin{example}
For $P=\mathtt{abac}$, $T=\mathtt{abcab}$, $I=[-2 \dd 3]$, and $z=4$, we have
\begin{align*}
\mu(P,T,I,z) &=\{(x,y)\,:\,P[x] \ne T[y],\,y-x \in I\}\\
&= \{(0,1),(0,2),\,(1,0),(1,2),(1,3),\,(2,1),(2,2),(2,4),\,(3,1),(3,3),(3,4)\};
\end{align*}
these are the top-left endpoints of missing diagonal edges in \cref{fig:align}. We further have $\mu(P,T,I,3) = \mu(P,T,I,4) \setminus \{(3,1),(3,3),(3,4)\}$.
\end{example}

We note that, in the next definition, the right endpoints of the intervals shrink by one unit.
This is because, for each $(x,y)\in M$, a bounding box of $M$ should contain both endpoints of the missing diagonal edge $(x,y) \stackrel{0}{\longleftrightarrow} (x+1,y+1)$.

\begin{definition}
    \label{def:bounding-box}
    For a finite set $M\subseteq \mathbb{Z}^2$, we say that
    $[p\dd p']\times [t\dd t']$ is a \emph{bounding box}
    of $M$ if $M\subseteq [p\dd p')\times [t\dd t')$.
\end{definition}

In what follows, we consider distances between points on the boundary of a bounding box.
We next define the \emph{left-top} and \emph{bottom-right} boundaries of a box; see \cref{fig:bbox} for an illustration.

\begin{definition}[{\cite[Definitions 8.6 and 8.7]{https://doi.org/10.48550/arxiv.2204.03087}}]\label{def:dmb}
    For a \emph{box} $B=[p\dd p']\times [t\dd t']$, we define the
    \emph{left-top} boundary \((\lt^B_d)_{d\in [t-p'\dd t'-p]}\)
    as a sequence of points with
    \[
        \lt^B_d \coloneqq \begin{cases}
            (t-d,t) & \text{for }d\in [t-p'\dd t-p],\\
            (p, p+d) & \text{for }d\in [t-p\dd t'-p].\\
      \end{cases}
    \]
    Further, we define the \emph{bottom-right boundary}
    \((\br^B_d)_{d\in [t-p'\dd t'-p]}\)
    as a sequence of points with
    \[
        \br^B_d \coloneqq \begin{cases}
            (p',p'+d) & \text{for }d\in [t-p'\dd t'-p'],\\
            (t'-d,t') & \text{for }d\in [t'-p' \dd t'-p].
        \end{cases}
    \]
    We then define an infinite 2-dimensional matrix $D_{M,B}$ as follows:
    \[D_{M,B}[a,b] := \begin{cases}
        \dist_M(\lt^B_a,\br^B_b) & \text{if }a,b\in [t-p'\dd t'-p],\\
        |a-b| & \text{otherwise}.
    \end{cases}
\]
\end{definition}

For an $s \times s$ matrix $A$, we define its
\emph{density matrix} $A^\square$ as an $(s-1) \times (s-1)$ matrix such that
\[
    A^\square [i,j] :=
    A[i+1,j] + A[i,j+1]
    - A[i,j] - A[i+1,j+1]
    \qquad\text{for \(i,j \in [0 \dd s-1)\).}
\]
This definition extends naturally to a (doubly) infinite 2D matrix $A$ (then, $A^\square$ is also doubly infinite).

\begin{definition}[{\cite[Definition 8.9]{https://doi.org/10.48550/arxiv.2204.03087}}]\label{def:dmpm}
For a finite set $M\subseteq \mathbb{Z}^2$, we define the \emph{general distance matrix} $D_M:= D_{M,B}$
for an arbitrary bounding box $B$ of $M$ and the \emph{seaweed matrix} $\Pi_M :=
\frac12 D_M^\square$.
\end{definition}

\begin{lemma}[{\cite[Lemmas 8.8 and 8.11]{https://doi.org/10.48550/arxiv.2204.03087}}]\label{lem:seaweed}
The general distance matrix $D_M=D_{M,B}$ does not depend on the choice of the bounding box $B$ of $M$.
Moreover, the seaweed matrix $\Pi_M$ is a bounded permutation matrix such that $D_M[a,b]=2\Pi_M^\Sigma[a,b]+a-b$ holds for all $a,b\in \mathbb{Z}$.
\end{lemma}

\begin{figure}
\begin{center}
\begin{tikzpicture}[scale=0.8]
\begin{scope}[xshift=-4cm]
\draw[densely dotted] (0,-6) grid (6,0);
\draw[-latex] (0,0) -- (0,-7);
\draw[-latex] (0,0) -- (7,0);
\draw[very thick] (2,-2) -- (3,-3) -- (3,-4) -- (5,-6);
\foreach \x in {1,...,6}{\draw (0,-\x) node[left] {\x};}
\foreach \x in {1,...,6}{\draw (\x,0) node[above] {\x};}
\draw (0,-7) node[below] {$x$};
\draw (7,0) node[right] {$y$};
\draw (2,-1) rectangle (5,-6);
\draw (5,-1) node[above right=-0.1cm] {4};
\draw (2,-1) node[above left=-0.1cm] {1};
\foreach \x/\y in {5/-1,2/-6}{\filldraw[purple] (\x,\y) circle (0.08cm);}
\foreach \x/\y/\v in {4/-1/-3,3/-1/-2,2/-1/-1,2/-2/0,2/-3/1,2/-4/2,2/-5/3}{\filldraw[red] (\x,\y) circle (0.08cm);}
\foreach \x/\y/\v in {4/-1/3,3/-1/2}{\draw (\x,\y) node[above=0.05cm,circle,draw=white,fill=white,inner sep=0pt] {\v};}
\draw (2,-6) node[below left=-0.1cm] {-4};
\draw (5,-6) node[below right=-0.1cm] {-1};
\foreach \x/\y/\v in {2/-2/0,2/-3/-1,2/-4/-2,2/-5/-3}{\draw (\x,\y) node[left=0.05cm,circle,draw=white,fill=white,inner sep=0pt] {\v};}
\foreach \x/\y/\v in {5/-2/-3,5/-3/-2,5/-4/-1,5/-5/0,5/-6/1,4/-6/2,3/-6/3}{\filldraw[blue] (\x,\y) circle (0.08cm);}
\foreach \x/\y/\v in {5/-2/3,5/-3/2,5/-4/1,5/-5/0}{\draw (\x,\y) node[right=0.05cm,circle,draw=white,fill=white,inner sep=0pt] {\v};}
\foreach \x/\y/\v in {4/-6/-2,3/-6/-3}{\draw (\x,\y) node[below] {\v};}
\draw[very thick,red,latex-] (5,-0.4) -- (1.4,-0.4) -- (1.4,-6);
\draw[very thick,blue,latex-] (5.8,-1) -- (5.8,-6.6) -- (2,-6.6);
\end{scope}

\begin{scope}[xshift=6.5cm, yshift=-6cm]
\draw[thick] (0.1,0) -- (0,0) -- (0,4.8) -- (0.1,4.8);
\draw[thick] (4.7,0) -- (4.8,0) -- (4.8,4.8) -- (4.7,4.8);
\begin{scope}[yshift=0.2cm]
\foreach \y/\l in {3.7/-3,3.2/-2,2.7/-1,2.2/0,1.7/1,1.2/2,0.7/3}{
    \filldraw[red] (-0.2,\y) circle (0.08cm);
    \draw (-0.2,\y) node[left] {\l};
}
\foreach \y/\l in {4.2/-4,0.2/4}{
    \filldraw[purple] (-0.2,\y) circle (0.08cm);
    \draw (-0.2,\y) node[left] {\l};
}
\end{scope}

\begin{scope}[xshift=0.2cm]
\foreach \x/\l in {3.7/3,3.2/2,2.7/1,2.2/0,1.7/-1,1.2/-2,0.7/-3}{
    \filldraw[blue] (\x,5) circle (0.08cm);
    \draw (\x,5) node[above] {\l};
}
\foreach \x/\l in {4.2/4,0.2/-4}{
    \filldraw[purple] (\x,5) circle (0.08cm);
    \draw (\x,5) node[above] {\l};
}
\end{scope}

\draw (1.9,2.4) node (tu) {1};
\draw[densely dotted] (-0.12,2.4) -- (tu) -- (1.9,4.92);

\draw (2.4,-0.3) node {distance matrix};

\draw[very thick,red,-latex,xshift=-0.6cm] (-0.3,4.8) -- (-0.3,0);
\draw[very thick,blue,-latex,yshift=0.6cm] (0,5.1) -- (4.8,5.1);
\end{scope}
\end{tikzpicture}
\end{center}
\caption{A box $[p \dd p'] \times [t \dd t']$ for $p=1$, $p'=6$, $t=2$, $t'=5$ and $D_{M,B}[t-p'\dd t'-p,\, t-p'\dd t'-p]$.
The entry $D_{M,B}[0,-1]$ stores the length of the shown shortest $(2,2)$-to-$(6,5)$ path (which is 1) since $(2,2)$ is the fifth vertex in the left-top boundary, $(6,5)$ is the fourth vertex in the bottom-right boundary, and $t-p' = -4$. (See also the numbering of the vertices of the bounding box.)
}\label{fig:bbox}
\end{figure}

The next lemma shows that the matrix $D_{q,z}$ considered before is a submatrix of the matrix $D_M$ for $M=\mu(P,T,[q-k \dd q+2k),z)$.

\SetKwFunction{Seaweed}{Seaweed}
\SetKwFunction{suc}{successor}
\SetKwFunction{push}{insert}
\SetKwFunction{swap}{swap}
\SetKwFunction{extractMin}{extractMin}
\newcommand{\Que}{\mathcal{Q}}
\newcommand{\spn}{\mathsf{span}}
\begin{lemma}\label{lem:distmatr}
Let $M=\mu(P,T,[q-k \dd q+2k),z)$, where $z \in [0 \dd m]$ and $q \in [0 \dd n)$.
For every $a,b \in [0 \dd 3k+1]$, we have
\[D_{q,z}[a,b]=D_{M}[a+q-k-1,b+q-k-1]=2\Pi_M^\Sigma[a+q-k-1,b+q-k-1]+a-b.\]
\end{lemma}
\begin{proof}
Let $B=[p \dd p'] \times [t \dd t']$ be a box with $p=0$, $p'=z$, $t=q-k-1$, $t'=z+q+2k$. According to \cref{def:dmb}, as $t'-t+1=z+3k+2$, we have
\begin{align*}
\lt^B_{t-p},\ldots,\lt^B_{t-p+3k+1} &= (p,t),\ldots,(p,t+3k+1),\\
\br^B_{t'-p'-3k-1},\ldots,\br^B_{t'-p'} &= (p',t'-3k-1),\ldots,(p',t').
\end{align*}
By plugging in formulas for $p,p',t,t'$, we obtain:
\begin{align*}
\lt^B_{q-k-1},\ldots,\lt^B_{q+2k} &= (0,q-k-1),\ldots,(0,q+2k),\\
\br^B_{q-k-1},\ldots,\br^B_{q+2k} &= (z,z+q-k-1),\ldots,(z,z+q+2k).
\end{align*}

Let $I=[q-k \dd q+2k)$. According to \cref{def:dmb}, for $a,b \in [0 \dd 3k+1]$, we have
\begin{align*}
    D_{M,B}[a+q-k-1,b+q-k-1] &= \dist_M(\lt^B_{a+q-k-1},\br^B_{b+q-k-1})\\
    &=\dist_M((0,a+q-k-1),(z,b+z+q-k-1)).
\end{align*}
By \cref{obs:useme}, $\AG(M)=G(P[0 \dd z),T,I)$, so $\dist_M((0,a+p-k-1),(z,b+z+q-k-1))$ is the length of the shortest path between the two points in $G(P[0 \dd z),T,I)$, which equals $D_{q,z}[a,b]$ by definition.

To conclude, we need to show that $D_{M,B}=D_M$, i.e., that $B$ is a bounding box of~$M$.
By definition, if $(x,y) \in M$, then $x \in [0 \dd z)$ and $y-x \in I$. 
Moreover, $I=[q-k \dd q+2k)$ implies that $y \in [q-k \dd z+q+2k)$.
Therefore, \[(x,y) \in [p \dd p') \times [t \dd t'),\] so $B=[p \dd p'] \times [t \dd t']$ satisfies \cref{def:bounding-box} for $M$. Thus, $D_{q,z}[a,b]=D_M[a+q-k-1,b+q-k-1]$.
Finally, note that \cref{lem:seaweed} implies 
\[D_M[a+q-k-1,b+q-k-1]\,=\,2\Pi_M^\Sigma[a+q-k-1,{b+q-k-1}]+(a+q-k-1)-(b+q-k-1)\]
\[=\,2\Pi_M^\Sigma[a+q-k-1,b+q-k-1]+a-b.\]
\end{proof}

An \emph{infinite permutation matrix} is a doubly infinite matrix that contains exactly one 1 in each row and column.
An infinite permutation matrix \(A\) is \emph{bounded} if the set $\{i\in \mathbb{Z} : A[i,i]=0\}$ is finite.

\begin{remark}
\cite[Lemma 8.11]{https://doi.org/10.48550/arxiv.2204.03087} shows that  for a finite set $M$, the seaweed matrix $\Pi_M$ is a bounded infinite permutation matrix.
\end{remark}

We say that a permutation $\sigma : [\ell \dd r) \to [\ell\dd r)$ \emph{represents} an infinite permutation matrix $A$ if 
\[A[i,j]=1 \Leftrightarrow (i\in [\ell\dd r) \text{ and } j = \sigma(i)) \text{ or }(i\notin [\ell\dd r)\text{ and } j = i).\]

For a non-empty finite subset $S\subseteq \mathbb{Z}$, we define
\(\spn(S) = [\min S \dd \max S]\); if $S$ is empty, $\spn(S)=\emptyset$. For a finite set $M\subseteq \mathbb{Z}^2$, we define $\spn(M)=\spn(\{y-x : (x,y)\in M\})$.

The proof of \cite[Lemma 8.23]{https://doi.org/10.48550/arxiv.2204.03087} is based on Algorithm~\ref{alg:seaweed} (\cite[Algorithm 4]{https://doi.org/10.48550/arxiv.2204.03087}).
For any $I:= (\ell\dd r)\supseteq \spn(M)$, Algorithm~\ref{alg:seaweed} maintains a permutation $\sigma : [\ell\dd r) \to [\ell\dd r)$
and implicitly iterates over all pairs $(p,d)\in \mathbb{Z}\times (\ell \dd r)$ in the lexicographic order.
The main technical result of \cite[Lemma 8.23]{https://doi.org/10.48550/arxiv.2204.03087} is encapsulated in the following statement.

\begin{lemma}[{\cite[proof of Lemma 8.23]{https://doi.org/10.48550/arxiv.2204.03087}}]\label{lem:claim}
Consider a finite set $M\subseteq \mathbb{Z}^2$.
Suppose that we are given an interval $I\supseteq \spn(M)$ and function
$\suc(M,\star,\star)$ that,
given $d\in I$ and $p\in \{-\infty\}\cup \mathbb{Z}$, returns
$\suc(M,d,p)=\min\{p'\ge p : (p',p'+d)\in M\}$ in $\Oh(1)$ time,
where $\min\emptyset=\infty$.
Algorithm~\ref{alg:seaweed} returns in $\Oh(|I|^2 \log \log |I|)$ time a permutation $\sigma^{-1}$ that represents $\Pi_M$.
\end{lemma}

\begin{algorithm}[htpb]
    \SetKwBlock{Begin}{}{end}
    $\Seaweed(I,\suc(M,\star,\star))$\Begin{
        $(\ell\dd r) \coloneqq  I$\;
        $\Que \coloneqq \emptyset$\;
        \lForEach{$d\in [\ell\dd r)$}{$\sigma(d):= d$}\label{ln:inisigma}
        \lForEach{$d\in (\ell\dd r)$}{$\Que.\push((\suc(M,d, -\infty),d))$}\label{ln:inique}
        \While{$\Que$ not empty}{
            $(p,d):=  \Que.\extractMin()$\;
            \If{$p\ne \infty$ \KwSty{and} $\sigma(d-1) < \sigma(d)$} {\label{line:if}
                $\swap(\sigma(d-1),\sigma(d))$\;
                \lIf{$d-1\in (\ell\dd r)$}{$\Que.\push((\suc(M,d-1, p+1),d-1))$}
                \lIf{$d+1\in (\ell\dd r)$}{$\Que.\push((\suc(M,d+1, p),d+1))$}
            }
        }
        \Return{$\sigma^{-1}$}\;
    }
    \caption{Constructing the seaweed matrix \cite[Algorithm 4]{https://doi.org/10.48550/arxiv.2204.03087}.}\label{alg:seaweed}
\end{algorithm}

We are now ready to prove the main result of this subsection, restated here for convenience.

\mainseaweed*
\begin{proof}
Given $q \in [0 \dd n)$, we provide an efficient construction algorithm of a permutation matrix~$\Pi_{q,z}$ for each $z \in [0 \dd m]$ that satisfies $D_{q,z}[i,j] = 2\Pi_{q,z}^{\Sigma}[i,j] + i-j$.

We use \cref{lem:claim} with $M := \mu(P,T,I,m)$ for $I:=[q-k\dd q+2k)$. Note that $I \supseteq \spn(M)$ by the definition of $M$.
Observe that, in this case, each value $\suc(M,d,p)$ can be implemented using a single $\LCP$ query on suffixes of $P$ and $T$.
Denote $p':= \max\{p,0,-d\}$. If $p' \le m$ and $p'+d \le n$, further define
\[P':=P[p' \dd m),\quad T':=T[p'+d \dd n),\quad\text{and}\quad \lambda:=\LCP(P',T').\]
Then
\[
\suc(M,d,p)=\left\{\begin{array}{ll}
\infty & \text{if }d \not\in I\text{, }p'>m\text{, }p'+d>n\text{, or }\lambda\in \{|P'|,|T'|\};\\
p' + \lambda & \text{otherwise.}
\end{array}\right.
\]

\begin{claim}
The permutation $\sigma^{-1}$ returned by Algorithm~\ref{alg:seaweed} represents a permutation matrix $\Pi_{q,m}$ such that $\forall_{i,j\in [0\dd 3k+1]}\; D_{q,m}[i,j]=2\Pi^{\Sigma}_{q,m}[i,j]+i-j$.
\end{claim}
\begin{proof}
Define a $(3k+1)\times (3k+1)$ matrix $\Pi_{q,m}:=\Pi_M[q-k-1\dd q+2k-1,q-k-1\dd q+2k-1]$. Since  $\sigma^{-1}:[q-k-1\dd q+2k)\to [q-k-1\dd q+2k)$ represents $\Pi_M$, we conclude that $\Pi_{q,m}$ is a permutation matrix represented by $\sigma^{-1}$. From \cref{lem:distmatr}, we have that, for every $a,b \in [0 \dd 3k+1]$,
\[D_{q,m}[i,j] = 2\Pi^\Sigma_M[i+q-k-1,j+q-k-1]+i-j=2\Pi^{\Sigma}_{q,m}[i,j]+i-j.\]
This completes the proof.
\end{proof}

Let $\ell := q-k-1$ and $r := q+2k$.
Algorithm~\ref{alg:seaweed} initializes an identity permutation on $[\ell \dd r)$ and processes elements of $M$ in the lexicographic order.
For each $(p,p+d)\in M$, the algorithm swaps $\sigma(d-1)$ and $\sigma(d)$ if $\sigma(d-1) < \sigma(d)$.

For an efficient implementation, the algorithm maintains a priority queue $\Que$ representing unprocessed elements of $M$.
The queue satisfies two invariants:
\begin{enumerate}
    \item For each $d\in (\ell\dd r)$, the queue $\Que$ stores at most one pair of the form $(\star,d)$: the pair $(p,d)$, where $(p,p+d)\in M$ is the next unprocessed element of $M$ on the diagonal $d$ of the general alignment graph $\AG(M)$, or $(\infty, d)$ if all elements of $M$ on diagonal $d$ have been processed.
    \item For each $d\in (\ell\dd r)$, the pair of the form $(\star,d)$ must be contained in $\Que$ if $\sigma(d-1) < \sigma(d)$.
\end{enumerate}
Initially, we have $\sigma(d-1)<\sigma(d)$ for all $d\in (\ell\dd r)$, so the algorithm inserts to $\Que$ pairs of the form $(\suc(M,d, -\infty),d)$ for every $d\in (\ell\dd r)$. 
Then, while the queue is not empty, the algorithm extracts the minimum element $(p,d)$ from $\Que$ and, if $p$ is finite and $\sigma(d-1) < \sigma(d)$, it
\begin{itemize}
\item swaps $\sigma(d-1)$ and $\sigma(d)$, and
\item inserts to the queue at most two pairs: for each $\delta \in \{d-1,d+1\}\cap (\ell \dd r)$,
roughly speaking, the next element of $M$ on the $\delta$-th diagonal.
\end{itemize}

Since the elements of $M$ are processed in the non-decreasing order with respect to their first coordinate, for each $z \in [0 \dd m]$, there is a time during the course of the algorithm when we have processed $\mu(P,T,I,z) = \{(p,p+d)\in  M : p < z\}$.
The same algorithm applied to $G(P[0 \dd z),T,I)$ with $M = \mu(P,T,I,z)$ would have performed the exact same swaps.
Hence, at that time, the permutation $\sigma^{-1}$ represents a permutation matrix $\Pi_{q,z}$ such that \[\forall_{i,j\in [0\dd 3k+1]}\; D_{q,z}[i,j]=2\Pi^{\Sigma}_{q,z}[i,j]+i-j.\]
5
In total, at most ${3k+1 \choose 2} = 3k(3k+1)/2$ swaps are performed:  this is because a swap is only performed if $\sigma(d-1) < \sigma(d)$, so any two values can get swapped at most once.

Recall that $\sigma$ represents a permutation matrix $\Pi$ such that $\Pi[i,j]=1$ if and only if $\sigma(i)=j$. In the same sense, $\sigma^{-1}$ represents the transposed permutation matrix $\Pi^T$.

If $\Pi$ is a permutation matrix of $\sigma$ and $\sigma'$ results by swapping $\sigma(d-1)$ and $\sigma(d)$, then the permutation matrix $\Pi'$ of $\sigma'$ is obtained from $\Pi$ by zeroing cells $\Pi[d-1,\sigma(d-1)]$ and $\Pi[d,\sigma(d)]$ and setting cells $\Pi[d-1,\sigma(d)]$ and $\Pi[d,\sigma(d-1)]$ to ones; effectively, $\Pi'$ is obtained by swapping rows $d-1$ and $d$ in $\Pi$. This corresponds to swapping neighbouring columns $d-1$ and $d$ of the transposed permutation matrix $\Pi^T$ of~$\sigma^{-1}$. 

If $\sigma(d-1)<\sigma(d)$, this is an ordered swap of columns. We thus indeed compute a sequence $\Delta(\Pi_{q,0},\Pi_{q,1}),\dots,\Delta(\Pi_{q,m-1},\Pi_{q,m})$ composed of at most $3k(3k+1)/2$ ordered swaps of columns.

The time complexity follows from the complexity of Algorithm~\ref{alg:seaweed}, the observation that function $\suc$ is called $\cO(|I|^2)$ times ($\cO(|I|)$ times before the while loop and at most two times for each of the  $\cO(|I|^2)$ performed swaps), and the fact that each evaluation of $\suc$ is reduced in constant time to an $\LCP$ query.
\end{proof}

\subsection{Main Result}
For a string $S$, by $S_\$$ we denote the string $S[0]\$S[1]\$ \cdots S[|S|{-}1]\$$.
By the following fact, we can easily transform the pattern and the text, doubling $k$, and henceforth consider the deletion distance instead of the edit distance.

\begin{fact}[{\cite[Section 6.1]{https://doi.org/10.48550/arxiv.0707.3619}}]\label{obs:ddist}
  For any two strings $U$ and $V$ over an alphabet $\Sigma$ that does not contain~$\$$,
  we have $2\cdot \ed(U,V)=\ddist(U_\$,V_\$)$.
\end{fact}

We start with the following auxiliary lemma.

\begin{lemma}\label{lem:allkPREFMATCH}
Consider a text $T$, a pattern $P$, and a positive integer $k$.
Given a position $q$ in $T$, the values $\LPref_{k'}[i]$ for all $k' \in [0 \dd k]$ and $i \in [q \dd q+k)\cap [0\dd n]$
can be computed in $\Oh(k^2 \log^2 k)$ time plus $\Oh(k^2)$ $\LCP$ queries on pairs of suffixes of $P$ and $T$.
\end{lemma}
\begin{proof}
It suffices to compute, for all $k' \in [0 \dd 2k]$ and $i \in [2q \dd 2q+2k)$, the length of the longest prefix of~$P_\$$ that matches a prefix of $T_\$[i \dd n)$ with deletion distance at most $k'$.
This is because \cref{obs:ddist} implies that, for each even $i$ and each even $k'$, the obtained length is equal to $2\cdot \LPref_{k'/2}[i/2]$, where $\LPref$ is defined for $P$ and $T$.

Let us replace $P$ with $P_\$$, $m$ by $2m$, $T$ with $T_\$$, $q$ by $2q$, and $k$ by $2k$ when using the tools from \cref{subsec:seaweeds}. 
Note that an $\LCP$ query on suffixes of $P_\$$ and $T_\$$ trivially reduces in $\cO(1)$ time to an $\LCP$ query on suffixes of $P$ and $T$.
We iterate over $D_{q,z}$, for all $z \in [0 \dd m]$, using the data structure of \cref{lem:Monge}.
The initialization of $D_{q,0}$ takes $\Oh(k^2)$ time; we have $D_{q,0}[a,b]=|a-b|$.
Each ordered swap of adjacent columns in the maintained matrix~$\Pi$ corresponds to a sub-column increment in $D$ (increasing the entries in the sub-column by~$2$).

Note that, for each $z \in [0 \dd m]$, $D_{q,z}$ is a Monge matrix due to \cref{obs:dt_monge}.
Each intermediate matrix~$D$ is also Monge as it satisfies $D[i,j] = 2\Pi^{\Sigma}[i,j]+i-j$ for the maintained permutation matrix~$\Pi$.
Thus, for each $z$ such that $\Delta(\Pi_{q,z},\Pi_{q,z+1}) \ne \emptyset$, we can update the maintained Monge matrix as necessary using \cref{lem:Monge} in $\Oh(\log^2 k)$ time per update.
By \cref{lem:interface:Kc}, the number of updates is $\Oh(k^2)$ and they can be computed in $\Oh(k^2\log \log k)$ time plus the time required to answer $\Oh(k^2)$ $\LCP$ queries.

For notational convenience, let us set $\min \emptyset = \infty$. Further, let
\[\gamma(i,z) := \min\{D_{q,z}[i-q+k+1,j-q+k+1]\,:\,j \in [q-k-1 \dd q+2k]
\cap [i-z\dd n-z]\}.\]
Then, for each $i \in [q \dd q+k)$ and for each $k' \in [0 \dd k]$, by \cref{lem:interface:Ka},
it suffices to compute the maximum $z \in [0 \dd m]$ such that $\gamma(i,z) \le k'$.

Observe that, for each $i \in [q \dd q+k)$, the intersection of $[q-k-1 \dd q+2k]$ and $[i-z\dd n-z]$ is neither empty nor equal to $[q-k-1 \dd q+2k]$ if and only if
\[z \in Z_i := [i-q-2k \dd i-q+k] \cup (n-q-2k \dd n-q+k+1].\]

We compute $\gamma(i,z)$ for each $i \in [q \dd q+k)$ and for each $z \in Z_i$ in total time $\cO(k^2 \log^2 k)$ using the $\Oh(\log^2 k)$-time sub-row minimum queries from \cref{lem:Monge} (by processing the $\Oh(k^2)$ updates in order).
We store the computed values in an $\cO(k) \times \cO(k)$ matrix $\Gamma$ in the natural order.

Then, for each pair of $i \in [q \dd q+k)$ and $k' \in [0 \dd k]$,
we use  binary search with $\Oh(\log k)$ iterations over the set of all $z$ that satisfy $\Delta(\Pi_{q,z},\Pi_{q,z+1}) \ne \emptyset$ to compute the sought maximum.
Whenever we need the value of $\gamma(i,z)$ for $z \in Z_i$, we simply read it from $\Gamma$---we can decide whether some $z$ is in $Z_i$ and, if so, compute the indices of the corresponding entry in~$\Gamma$ in $\cO(1)$ time using standard arithmetic.
Otherwise, that is, when we consider some $z \not\in Z_i$, the interval $[q-k-1 \dd q+2k] \cap [i-z\dd n-z]$ is either empty or equal to $[q-k-1 \dd q+2k]$.
In the former case, $\gamma(i,z)>k'$.
In the latter case, $\gamma(i,z)$ is equal to the minimum value in row $(i-q+k+1)$ of~$D_{q,z}$ and can be thus computed in $\Oh(\log k)$-time using a row minimum query from \cref{lem:Monge}.
The performed binary searches thus take $\Oh(k^2\log^2 k)$ time in total.
\end{proof}

We are now ready to prove \cref{thm:allkPREFMATCH}, restated here for convenience.

\allkPREFMATCH*
\begin{proof}
We invoke the algorithm underlying \cref{lem:allkPREFMATCH} for each $q \in [0 \dd n)$ that satisfies $q \equiv 0 \pmod{k}$ noting that $\LCP$ queries on the suffixes of $P$ and $T$ can be answered in constant time after an $\cO(n)$-time preprocessing (cf.\ \cref{the:standardPILLAR}).
The overall running time is thus $\Oh(n+ (n/k) \cdot k^2\log^2 k)=\Oh(nk\log^2 k)$.
\end{proof}

\section{Conditional Hardness of Approximate CPM}\label{sec:hardness}
We consider the following problem in which the number of allowed mismatches is unbounded.

\defproblem{Mismatch-CPM}{
  A text $T$ of length $n$ and a pattern $P$ of length $m$.
}{
  An array $\CPMA[0\dd n-m]$ with $\CPMA[i]\,=\,  \min\{k \ge 0\,:\,P \approx_k^{\Ham} T[i \dd i+m)  \}$.
}

In jumbled indexing, we are to answer pattern matching queries in the jumbled (Abelian) sense.
More precisely, given a Parikh vector of a pattern that specifies the quantity of each letter, we are to check if there is a substring of the text
with this Parikh vector.
In the case of a binary text, the problem of constructing a jumbled index is known to be equivalent (up to a $\log n$-factor in the case where a witness occurrence needs to be identified; see~\cite{DBLP:conf/spire/CicaleseGGLLRT13})
to the following problem:

\newcommand{\mmin}{\mathit{min}}
\newcommand{\mmax}{\mathit{max}}
Given a text $X$ of length $n$ over alphabet $\{0,1\}$, for all $t \in [1 \dd n]$ compute the values:
  \[
      \mmin_t := \min\left\{\sum_{i=p}^{p+t-1} X[i]\,:\, p \in [0 \dd n-t]\right\},\,
  \mmax_t := \max\left\{\sum_{i=p}^{p+t-1} X[i]\,:\, p \in [0 \dd n-t]\right\}.
  \]

For a few years since its introduction~\cite{DBLP:conf/stringology/CicaleseFL09}, the problem of constructing a binary jumbled index (BJI) in $\Oh(n^{2-\varepsilon})$ time for $\varepsilon>0$ was open.
Chan and Lewenstein~\cite{DBLP:conf/stoc/ChanL15} settled this question affirmatively by proposing an $\Oh(n^{1.859})$-time randomized construction; very recently, it was improved to $\Oh(n^{1.5}\log^{\Oh(1)} n)$ time~\cite{https://doi.org/10.48550/arxiv.2204.04500}.
We make the following reduction.

\begin{theorem}\label{hard:mismatch}
  If the Mismatch-CPM problem on binary strings can be solved in $S(n)$ time, then a BJI can be constructed in $\Oh(n+S(3n))$ time.
\end{theorem}
\begin{proof}
  We show how to compute $\mmax_t$ for all $t \in [1 \dd n]$.
  For computing $\mmin_t$, we can negate all the letters of $X$. 
  An illustration of our reduction is provided in \cref{fig:AppCPM}.
  
  \begin{figure}[htpb]
  \centering
\includegraphics[width=14cm]{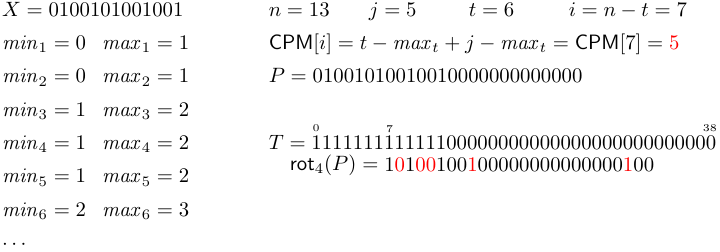}
\caption{$\CPMA[7]=5$ corresponds to a Hamming distance of $5$ between $T[7\dd 7+|P|)$ and some rotation of $P$, namely of $\rot_4(P)$.}\label{fig:AppCPM}
\end{figure}

  We construct an instance of the Mismatch-CPM problem with $P=X 0^n$ and $T=1^n0^{2n}$.
  For each length $t\in [0\dd n]$, the value $\CPMA[n-t]$ represents the minimum Hamming distance between $1^t 0^{2n-t}$ and a cyclic rotation $\rot_r(P)$ of $P$ with $r\in \mathbb{Z}$.
  In this case, the mismatches correspond to~$0$s in $P^\infty[r\dd r+t)$ and $1$s in $P^\infty[r+t\dd r+2n)$.
  The total number of $1$s in $P^\infty[r\dd r+2n)$ equals $j:=\sum_{i=0}^{n-1} X[i]$ independently of $r$, so 
  \[\Ham(1^t 0^{2n-t}, \rot_r(P)) = t+j-2\sum_{i=r}^{r+t-1} P^\infty[i].\]
  Thus, $(t+j-\CPMA[n-t])/2$ is equal to the maximum number of $1$s in any length-$t$ cyclic substring of~$P$.
  This maximum is attained at a length-$t$ substring of $X$ because cyclic substrings of $P$ disjoint with $X$ do not have any $1$s, and cyclic substrings partially overlapping $X$ can be shifted so that the number of $1$s does not decrease and the overlap with $X$ increases.
  Hence, all the values $\mmax_t=(t+j-\CPMA[n-t])/2$ for $t\in [0\dd n]$ can be recovered from the \CPMA\ array in $\Oh(n)$ time.
\end{proof}

The following theorem shows how to compute the edit distance of two strings with the use of an algorithm for the decision version of $k$-Edit CPM, and thus also that a strongly subquadratic algorithm for this problem would refute SETH~\cite{DBLP:journals/jcss/ImpagliazzoP01,DBLP:journals/siamcomp/BackursI18,DBLP:conf/focs/BringmannK15}.

\begin{restatable}{theorem}{hardedit}\label{hard:edit}
  If the $k$-Edit CPM problem on quaternary strings can be solved in $\Oh(n^{2-\varepsilon})$ time for some constant $\varepsilon >0$, then the edit distance of two binary strings each of length at most $n$ can be computed in $\Oh(n^{2-\varepsilon}\log n)$ time.
\end{restatable}

\begin{proof}
 To compute the edit distance between the binary strings $U$ and $V$, each of length at most $n$, it suffices to binary search on $k$, solving $k$-Edit CPM for pattern $P=\$^{3n} U \#^{3n}$ and text $T=\$^{3n} V \#^{3n}$, where the letters $\$$ and $\#$ do not occur in $UV$.
 Thus, $P$ and $T$ are quaternary strings.
 
 Let us denote \[\ED(i,j,p)\,=\, \ed(\rot_{j\bmod |P|}(P),T[i\dd p)),\] it is the edit distance between the $(j\bmod |P|)$-th rotation\footnote{Here, we consider $j\bmod |P|$ instead of just $j$ to account for the case where $j<0$.} of $P$ and the length-$(p-i)$ \fragment of $T$ starting at position $i$ of $T$. Let us notice that $\ED(0,0,|T|)=\ed(P,T)=\ed(U,V)$ since an insertion or a deletion of the same letter at the start or the end of both strings does not change their edit distance.

 \begin{claim}\label{clm:ed_red}
 $\ED(i,j,p)$ is minimized when $i=0$, $j=0$, and $p=|T|$.
 \end{claim}
 
\begin{proof}
  Notice that, for $p-i<5n$, we have $\ED(i,j,p)>n\ge \ed(U,V)$ since the difference in lengths of $\rot_{j\bmod |P|}(P)$ and $T[i\dd p)$ is greater than $n$. From now on, we assume that $i\le 2n$ and $p\ge |T|-2n$.
  The proof proceeds by induction on $|j|$ for $j\in (3n-|P|\dd 3n]$.
  
  In the base case ($j=0$), we can remove the longest common prefix (common $\$$ letters) and the longest common suffix (common $\#$ letters) from the prefix and suffix of both strings without changing the edit distance.
  We are left with $V$ as the text (since only $\$$ and $\#$ are missing from~$T$ due to the length assumption), and $\$^*U\#^*$ as the pattern. All those letters $\$$ and $\#$ have to be removed or substituted since there is no letter matching those in $V$. Any such removal increases the number of edit operations over the case with a smaller number of such letters, while any substitution can be simulated with an insertion in such a case; hence, it is optimal to have the least number of such letters, which is obtained for $i=0$ and $p=|T|$; namely, in this case, the length of the longest common prefix (consisting of $\$$ letters) and suffix (consisting of $\#$ letters) is maximized.
  
  Assume now that $0<j\le 3n$ so that $\rot_{j\bmod |P|}(P)$ ends with a letter $\$$. This letter cannot be matched in $T[i \dd p)$ since $|P|>6n$, and the last $\$$ in $T[i \dd p)$ appears before position $3n$; thus, it has to be either deleted or substituted with the last letter of $T[i \dd p)$. If we delete it, then the resulting string is the same as if we had deleted the first letter of $\rot_{(j-1)\bmod |P|}(P)$. 
  If we substitute it with $\#$ and then remove the last matching letters from both strings, we obtain the same result as if we had deleted the first letter from $\rot_{(j-1)\bmod |P|}(P)$ and computed the distance with $T[i\dd p-1)$. Hence, 
  \[\ED(i,j,p)\ge\min (\ED(i,j-1,p),\ED(i,j-1,p-1)),\] which is larger than or equal to $\ED(0,0,|T|)$ by the induction hypothesis.
  
  The remaining case is symmetric to the previous one: Assume now that $3n-|P|<j<0$ so that $\rot_{j\bmod |P|}(P)$ starts either with $\#$ (if $j\ge -3n$) or with a letter from $U$ (otherwise). In both cases, this letter cannot be matched with a letter from $T[i \dd p)$ since $i\le 2n$ (hence $\$$ appears on the first $n$ positions there).
  If we delete this letter, then we would obtain the same result as if we had deleted the last letter from $\rot_{(j+1)\bmod |P|}(P)$. If we substitute it with $\$$ and then remove the matching first letters from both strings, we obtain the same result as if we had deleted the last letter from $\rot_{(j+1)\bmod |P|}(P)$ and matched it with $T[i+1\dd p)$. Hence, $\ED(i,j,p)\ge\min (\ED(i,j+1,p),\ED(i+1,j+1,p))$, which is greater than or equal to $\ED(0,0,|T|)$ by the induction hypothesis.
 \end{proof}

  By \cref{clm:ed_red}, $k$-Edit CPM will not find any occurrence for any $k<\ed(U,V)$, and hence the binary search will result in finding exactly $\ed(U,V)$.
\end{proof}

\section*{Acknowledgments}

We thank Pawe\l{} Gawrychowski and Oren Weimann for helpful discussions related to \cref{lem:Monge}.

Panagiotis Charalampopoulos was supported by Israel Science Foundation (ISF) grant 810/21 when part of the work that led to this paper was conducted.

Jakub Radoszewski was supported by the Polish National Science Center, grants no.\ 2018/31/D/ST6 /03991 and 2022/46/E/ST6/00463. He was affiliated also at Samsung R\&D Warsaw, Poland when working on this manuscript.

Solon P.~Pissis was supported by the PANGAIA and ALPACA projects that have received funding from the European Union’s Horizon 2020 research and innovation programme under the Marie Skłodowska-Curie grant agreements No 872539 and 956229, respectively.

Tomasz Waleń was supported by the Polish National Science Center, grant no.\ 2018/31/D/ST6 /03991. 

Wiktor Zuba was supported by the Netherlands Organisation for Scientific Research (NWO) through Gravitation-grant NETWORKS-024.002.003.

\bibliographystyle{plainurl}
\bibliography{references}

\appendix
\section{Exact Circular Pattern Matching}\label{sec:CPM}

In this section, for completeness, 
we present a linear-time solution for CPM \emph{for any alphabet}.\footnote{This solution may be regarded as folklore, but we were unable to find a reference.}

\begin{fact}
Exact CPM can be solved in $\cO(n)$ time.
\end{fact}
\begin{proof}
We use the so-called \emph{prefix array} $\PREF$ of a string $S$ such that,  for $i \in [0\dd |S|)$, the value $\PREF[i]$ is the maximum $j\ge 0$ such that $S[0\dd j)=S[i\dd i+j)$.
An algorithm constructing the prefix array can be found in~\cite{DBLP:books/daglib/0020103}.
It works in linear time for any alphabet that supports constant-time equality tests between characters of $S$.

Assume that we have precomputed the prefix array  for a string $P\$T$ and the prefix array for a string $P^R\$T^R$, where $P^R$ and $T^R$ are the reversals of strings $P$ and $T$.
Then for each position $p$ in~$T$ we can compute 
$\mathit{pref}[p]$ ($\mathit{suf}[p]$), defined as the length of the longest prefix (resp., suffix)  of 
$P$ which is a prefix of $T[p\dd n)$ (resp., a suffix of $T[0 \dd p)$).
Let \[\mathcal{I}_p=[p-\mathit{suf}[p]\dd p+\mathit{pref}[p]-m].\]
Then the set of positions in~$T$ where some rotation of $P$ occurs is
the union of all intervals $\mathcal{I}_p$ for $p\in [0\dd n)$.
We compute the intervals' union in $\Oh(n)$ time using radix sort.
\end{proof}

\end{document}